\documentclass[aps,prx,footinbib,superscriptaddress,floatfix]{revtex4-2}
\usepackage[utf8]{inputenc} % allow utf-8 input
\usepackage[T1]{fontenc}    % use 8-bit T1 fonts
\usepackage[hidelinks]{hyperref}       % hyperlinks
\usepackage{url}            % simple URL typesetting
\usepackage{booktabs}       % professional-quality tables
\usepackage{amsfonts}       % blackboard math symbols
\usepackage{nicefrac}       % compact symbols for 1/2, etc.
\usepackage{microtype}      % microtypography
\usepackage{xcolor}         % colors
\usepackage{amsmath}
\usepackage{amssymb}
\usepackage{mathtools}
\usepackage{amsthm}
\usepackage{bbm}
\usepackage{wrapfig}
\usepackage{epstopdf}
\usepackage[pdftex]{graphicx}
\usepackage[capitalize,noabbrev]{cleveref}
\usepackage{subfig}  % Or \usepackage{subcaption} for newer versions

\usepackage{algorithm}
\usepackage{algorithmic}

\definecolor{darkblue}{rgb}{0,0,0.5}
\hypersetup{
    colorlinks=true,
    linkcolor=black,
    filecolor=blue,
    citecolor=darkblue,  
    urlcolor=black,
}
\newtheorem{theorem}{Theorem}

\newtheorem{definition}{Definition}
\newtheorem{remark}{Remark}

\def\be{\begin{equation}}
\def\ee{\end{equation}}
\def\ba{\begin{eqnarray}}
\def\ea{\end{eqnarray}}

\begin{document}
%TC:ignore
\title{Overcoming Dimensional Factorization Limits in Discrete Diffusion Models through Quantum Joint Distribution Learning}

\author{Chuangtao Chen}
\email{chuangtaochen@gmail.com}
\affiliation{Faculty of Innovation Engineering, Macau University of Science and Technology, Macao 999078, China
}

\author{Qinglin Zhao\footnote{Corresponding author: qlzhao@must.edu.mo}}
% \email{qlzhao@must.edu.mo}
\affiliation{Faculty of Innovation Engineering, Macau University of Science and Technology, Macao 999078, China
}

\author{MengChu Zhou}
% \email{xiaohuic@usc.edu}
\affiliation{Department of Electrical and Computer Engineering, New Jersey Institute of Technology, Newark, NJ 07102 USA}

\author{Dusit Niyato}
% \email{qzhuang@usc.edu}
\affiliation{
	College of Computing and Data Science, Nanyang
	Technological University, 639798 Singapore
} 
\author{Zhimin He}
% \email{qzhuang@usc.edu}
\affiliation{
School of Electronic and Information Engineering, Foshan University, Foshan 528000, China
}
% \affiliation{ Department of Physics and Astronomy, University of Southern California, Los
% Angeles, California 90089, USA
% }
\author{Haozhen Situ}
% \email{qzhuang@usc.edu}
\affiliation{
College of Mathematics and Informatics, South China Agricultural University, Guangzhou 510642, China
}

\begin{abstract}
	Discrete diffusion models represent a significant advance in generative modeling, demonstrating remarkable success in synthesizing complex, high-quality discrete data. 
	However, to avoid exponential computational costs, they typically rely on calculating per-dimension transition probabilities when learning high-dimensional distributions. In this study, we rigorously prove that this approach leads to a worst-case linear scaling of Kullback-Leibler (KL) divergence with data dimension.
	To address this, we propose a Quantum Discrete Denoising Diffusion Probabilistic Model (QD3PM), which enables joint probability learning through diffusion and denoising in exponentially large Hilbert spaces, offering a theoretical pathway to faithfully capture the true joint distribution.
	By deriving posterior states through quantum Bayes' theorem, similar to the crucial role of posterior probabilities in classical diffusion models, and by learning the joint probability, we establish a solid theoretical foundation for quantum-enhanced diffusion models. 
	For denoising, we design a quantum circuit that utilizes temporal information for parameter sharing and incorporates learnable classical-data-controlled rotations for encoding.
    Exploiting joint distribution learning, our approach enables single-step sampling from pure noise, eliminating iterative requirements of existing models. 
    Simulations demonstrate the proposed model's superior accuracy in modeling complex distributions compared to factorization methods. Hence, this paper establishes a new theoretical paradigm in generative models by leveraging the quantum advantage in joint distribution learning.
\end{abstract}
\maketitle
%TC:endignore

\twocolumngrid

\section{Introduction}
\label{introduction}

Generative learning, a key area in deep learning, focuses on training models to learn the probability distribution of a given dataset \cite{Kingma2013,Lecun2015,Goodfellow2020,Bond2021,chen2023mogan}. Once trained, the model is capable of generating new samples that follow this distribution. Recently, denoising diffusion probabilistic models  \cite{Sohl2015,Ho2020} have garnered significant attention due to their stable training processes and impressive generation results. Several important variants \cite{ho2022video,rombach2022high,croitoru2023diffusion} have been developed. Discrete data is a crucial type in data science, and several discrete diffusion models \cite{austin2021structured,hoogeboom2021argmax,campbell2022continuous} have been proposed to handle them. Discrete diffusion models, which add noise to discrete data to transform the original data distribution into a stationary distribution (e.g., a uniform distribution), have achieved remarkable performance in text generation \cite{hoogeboom2021argmax,he2023diffusionbert}, image generation \cite{austin2021structured,Hu_2022_CVPR}, video generation \cite{ref217}, and graph generation \cite{kong2023autoregressive,chen2024d4explainer}.

\begin{figure}[t]
	\centering
	\includegraphics[width=0.8\linewidth]{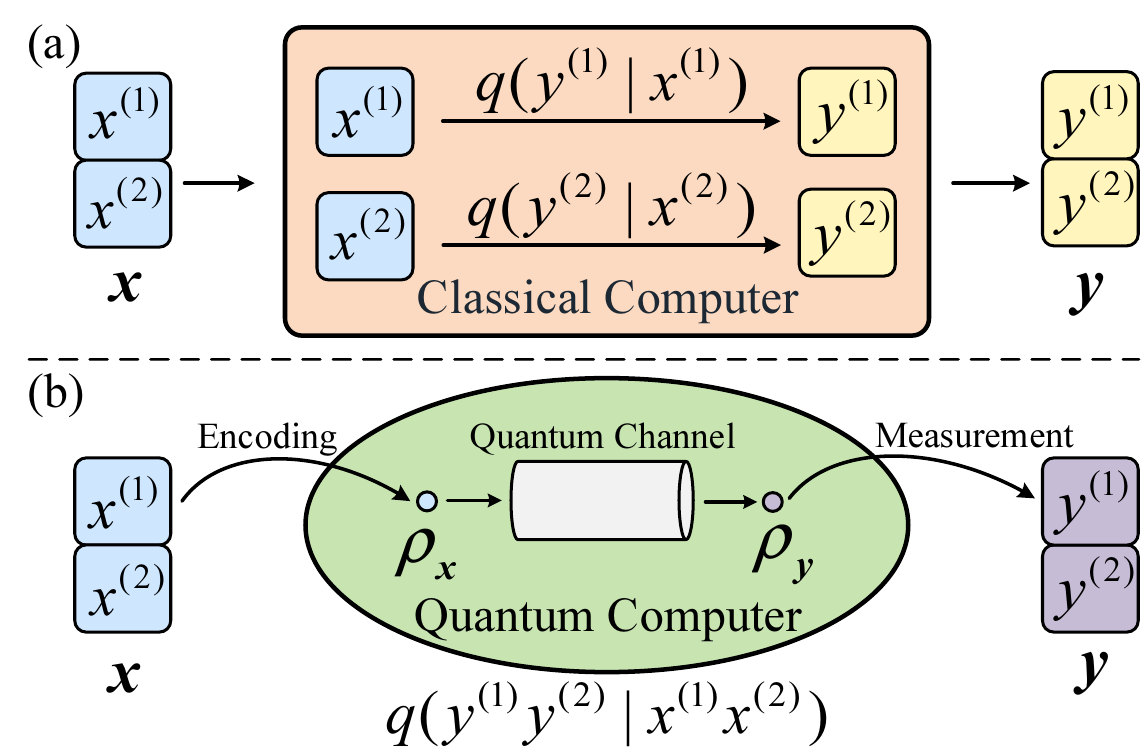}
	\caption{Comparison of classical and quantum approaches for data processing. (a) The classical model independently computes transition probabilities for each data dimension. (b) The quantum model uses quantum states and channels to model joint distributions.}
	\label{fig:core idea}
\end{figure}

While these existing approaches demonstrate significant progress, quantum computing presents a new paradigm for generative learning through enhanced computational capabilities and novel algorithms. Quantum generative models \cite{tian2023recent} leverage quantum computing \cite{nielsen2010quantum} to explore these opportunities and to introduce innovative frameworks including quantum generative adversarial networks \cite{lloyd2018quantum, dallaire2018quantum,chakrabarti2019quantum,zoufal2019quantum,situ2020quantum,huang2021quantum,niu2022entangling,silver2023mosaiq,kim2024hamiltonian,ma2025quantum}, quantum denoising diffusion probabilistic models (QDDPMs) \cite{parigi2024quantum,zhang2024generative,chen2024quantum,de2024quantuma,de2024quantumb,kwun2024mixed} and others \cite{amin2018quantum,liu2018differentiable,khoshaman2018quantum, benedetti2019generative,zoufal2021variational}. 
% Notably, QDDPMs employ a multi-step denoising process that decomposes complex generation tasks into simpler subproblems, potentially offering improved fitting performance over direct distribution generation approaches. 
These quantum generative models open new avenues for data generation and offer the potential to enhance classical generative models by exploiting quantum characteristics such as superposition and correlations \cite{gao2022enhancing}.

\textbf{Motivations.} 
%Despite recent advancements, existing discrete diffusion models suffer from inherent limitations. As illustrated in Figure 1(a), the diffusion process computes transition probabilities independently across dimensions. Likewise, the ground-truth posterior used during training is estimated separately for each dimension. While this approach avoids the exponential complexity of joint distribution modeling, it neglects inter-dimensional correlations, leading to two critical issues. The loss of dependency modeling degrades the quality of distribution fitting, as Liu et al. [39] qualitatively demonstrate through the emergence of an irreducible term in the Evidence Lower Bound (ELBO), which fundamentally restricts the model’s ability to match the target distribution. Additionally, preserving sampling quality under these constraints requires a large number of denoising iterations [40], thereby significantly reducing computational efficiency.
Despite these mentioned advancements, existing discrete diffusion models face fundamental limitations. As shown in \cref{fig:core idea}(a), their diffusion process independently computes transition probabilities for each dimension. Similarly, when computing the ground-truth posterior for training, they calculate the posterior distribution for each dimension separately.
While these strategies avoid the exponential costs of joint distribution modeling, they inevitably sacrifice inter-dimensional correlations, resulting in two significant limitations. Primarily, they degrade distribution fitting quality. Liu et al. \cite{liu2025discrete} qualitatively show that these methods introduce an irreducible term in the Evidence Lower Bound (ELBO), fundamentally limiting the model's ability to match the target distribution. Furthermore, these approaches necessitate an extensive number of denoising iterations to maintain sampling quality \cite{hayakawa2024distillation}, thus significantly reducing computational efficiency.

These limitations become particularly pronounced when confronting high-dimensional discrete data, where the curse of dimensionality renders joint distribution modeling intractable. 
This creates a moment for quantum computing, whose unique, non-classical characteristics make it naturally suited to address the challenges in discrete diffusion models, particularly through two key features: 1) its exponentially large Hilbert space allows for the direct learning of joint distributions over high-dimensional discrete data, and 2) quantum measurements provide a native mechanism for generating discrete outcomes, a natural fit for discrete data generation.
% This creates a moment for quantum computing, whose exponentially large Hilbert space is naturally suitable for encoding high-dimensional joint distributions. 
To realize an efficient fully-quantum discrete diffusion model, we have to perform four critical steps: 1) designing a quantum-based diffusion process for joint distributions, 2) deriving posterior distributions suited to this quantum framework, as the existing approach operates on classical computers and calculating posteriors independently for each dimension, whereas the quantum framework computes posterior distributions on quantum computers, 3) designing quantum circuit architecture for effective denoising, and 4) establishing a one-step inference framework for efficient generating clean data, bypassing multiple denoising iterations. Consequently, we can well establish the foundation for quantum-enhanced discrete diffusion models.

\textbf{Contributions.}
To address these, we propose a Quantum Discrete Denoising Diffusion Probabilistic Model (QD3PM), as the first quantum diffusion framework that fundamentally overcomes dimension-wise factorization constraints through joint distribution modeling in Hilbert space. 
As shown in \cref{fig:core idea}(b), our method implements quantum diffusion by: 1) encoding data into quantum states, 2) evolving through quantum channels while preserving inter-dimensional correlations, and 3) sampling new instances from the joint probability distribution via quantum measurement.
Specifically, we make the following novel contributions to the field of discrete diffusion models:
                     
\textbf{1. Theoretical characterization of existing discrete diffusion limitations.} We theoretically prove that the worst-case fitting error of existing discrete diffusion models for $N$-dimensional data with $K$ possible values per dimension scales as $\mathcal{O}(N\log K)$, due to the inherent dimension-factorization approach.
Remarkably, our quantum-enhanced QD3PM circumvents this limitation through joint distribution modeling in exponentially scaled Hilbert space, achieving theoretically perfect distribution matching.

\textbf{2. Joint distribution modeling in Hilbert space with quantum Bayes' theorem for accurately capturing inter-dimensional correlations.} We develop a method to leverage the exponentially large Hilbert space for storing and computing joint probability distributions in the diffusion process. Using quantum Bayes' theorem \cite{leifer2013towards}, we derive a general formulation for the posterior distribution state whose measurement distribution corresponds to the ground-truth posterior. We further provide a specific derivation for the posterior state when the diffusion process is modeled as a depolarizing channel, along with its quantum circuit implementation.

\textbf{3. Temporal-aware denoising architecture for minimizing parameter overhead.} We carefully design a circuit architecture for the QD3PM's denoising process. This circuit incorporates timestep encoding for parameter sharing across different timesteps, preventing the need for step-specific parameter sets across the denoising process, thereby minimizing parameter overhead. Additionally, it integrates learnable rotation gates to encode classical information into quantum states, ensuring effective denoising even from a uniform initial distribution. Therefore, this parameter-efficient circuit overcomes the fundamental invariance of the maximally mixed state to unitary evolution.

\textbf{4. One-step inference framework for achieving faster generation.}
We develop an approach that enables our model to generate high-quality samples in a single step after training by leveraging fundamental properties of quantum operations and measurements. This breakthrough allows our model to bypass the numerous denoising iterations required by existing discrete diffusion models, dramatically improving sampling efficiency compared to classical methods.

Our work goes beyond replacing existing quantum generative models; it integrates the diffusion model's framework with quantum computing. We tackle key challenges of existing discrete diffusion models, i.e., dimension disaster and correlation breakdown while laying the theoretical groundwork for quantum diffusion models and showcasing the potential of quantum computing to elevate existing machine learning. 
%Our work is not a simple replacement of existing quantum generative models, but rather a deep integration of the diffusion model's mathematical framework with the advantages of quantum computing. This paper addresses the fundamental bottlenecks of classical discrete diffusion models (dimension disaster and correlation breakdown), while establishing a theoretical foundation for quantum diffusion models and demonstrating the potential of quantum computing to enhance classical machine learning models.

\section{Results}
\subsection{Preliminary}
\label{section: preliminary}

%Discrete diffusion models \cite{hoogeboom2021argmax,austin2021structured,campbell2022continuous,chenanalog,santos2023blackout,loudiscrete,hayakawa2024distillation,liu2025discrete} are a class of diffusion models used to generate discrete data.
%Among them, 
Discrete Denoising Diffusion Probabilistic Model (D3PM) \cite{austin2021structured} is a representative discrete diffusion model. Its forward diffusion process transforms an $N$-dimensional data \(\boldsymbol{x}_0 = \{x_0^{(1)}, x_0^{(2)}, \ldots, x_0^{(N)}\}\) into noisy data \(\boldsymbol{x}_T = \{x_T^{(1)}, x_T^{(2)}, \ldots, x_T^{(N)}\}\) through discrete timesteps $t = 1, 2, \ldots, T$, with $\boldsymbol{x}_T$ following a complete noise distribution. 
% The superscript ``$ (i) $\textquotedblright~ indicates that \( x_t^{(i)} \) is the \( i \)-th component of \( \boldsymbol{x}_t \). 
In D3PM, the transition probability matrix $\boldsymbol{Q}_t$ acts independently on each dimension of the data (e.g., pixels in an image or tokens in a sequence), avoiding exponential increases in storage and computation but breaking inter-dimensional correlations. Each scalar random variable $x_t^{(i)}$ has $K$ categories and is represented by a one-hot vector $\boldsymbol{x}_t^{(i)}$. The transition probability at timestep $t$ is:
\begin{equation}
	q(\boldsymbol{x}_t^{(i)} | \boldsymbol{x}_{t-1}^{(i)}) = \text{Cat}(\boldsymbol{x}_t^{(i)}; \boldsymbol{p} = \boldsymbol{x}_{t-1}^{(i)} \boldsymbol{Q}_t),
\end{equation}
where $\text{Cat}(\boldsymbol{x}_t^{(i)}; \boldsymbol{p})$ is a categorical distribution over the one-hot encoded vector  $ \boldsymbol{x} $ with probabilities $\boldsymbol{p}  $, and $\boldsymbol{x}_{t-1} \boldsymbol{Q}_t$ represents the standard matrix-vector multiplication. The probability transition matrix $\boldsymbol{Q}_t$ has dimensions $K \times K$. It can be constructed as
$\boldsymbol{Q}_t = \alpha_t \boldsymbol{I} + \frac{1 - \alpha_t}{K} \mathbbm{1}\mathbbm{1}^T$,
which defines a uniform transition matrix \cite{hoogeboom2021argmax},
where $\boldsymbol{I}$ is the identity matrix, $\alpha_t \in [0,1]$, and $\mathbbm{1}$ is a $K$-dimensional column vector with all entries equal to one. The relationship between $\boldsymbol{x}_0^{(i)}$ and $\boldsymbol{x}_t^{(i)}$ at timestep $t$ is given as:
\begin{equation}
	q(\boldsymbol{x}_t^{(i)} \mid \boldsymbol{x}_0^{(i)}) = \text{Cat}\left(\boldsymbol{x}_t^{(i)}; \boldsymbol{p} = \boldsymbol{x}_0^{(i)} \overline{\boldsymbol{Q}}_t \right),
\end{equation}
where $\overline{\boldsymbol{Q}}_t = \boldsymbol{Q}_1 \boldsymbol{Q}_2 \cdots \boldsymbol{Q}_t$.
The denoising process in discrete diffusion models is parameterized by $p_{\boldsymbol{\theta}}(\boldsymbol{x}_{t-1}^{(i)} | \boldsymbol{x}_t^{(i)})$. In order to fit the data distribution $ q(\boldsymbol{x}_0) $, the model is trained by minimizing the variational lower bound objective function overall dimensions:
\begin{equation} 
	\begin{split} 
		\label{eq:lower bound}
		\mathcal{L}_{vb} &= \mathbb{E}_{q(\boldsymbol{x}_0)} \Bigg[ 
		\underbrace{D_{KL}\big[q(\boldsymbol{x}_T | \boldsymbol{x}_0) \| p(\boldsymbol{x}_T)\big]}_{\mathcal{L}_T} \\
		&+ \sum_{t=2}^T \mathbb{E}_{q(\boldsymbol{x}_t | \boldsymbol{x}_0)} [ 
		\underbrace{D_{KL}\big[q(\boldsymbol{x}_{t-1} | \boldsymbol{x}_t, \boldsymbol{x}_0) \| p_\theta(\boldsymbol{x}_{t-1} | \boldsymbol{x}_t)\big]}_{\mathcal{L}_{t-1}} ] \\
		& - \underbrace{\mathbb{E}_{q(\boldsymbol{x}_1 | \boldsymbol{x}_0)}\big[\log p_\theta(\boldsymbol{x}_0 | \boldsymbol{x}_1)\big]}_{\mathcal{L}_0}
		\Bigg],
	\end{split} 
\end{equation} 
where the ground-truth posterior distribution of the $ i $-th dimension is obtained via Bayes' theorem:
\begin{equation} 
	\begin{split} 
		&q(\boldsymbol{x}_{t-1}^{(i)} | \boldsymbol{x}_t^{(i)}, \boldsymbol{x}_0^{(i)})  
		= \frac{q(\boldsymbol{x}_t^{(i)} | \boldsymbol{x}_{t-1}^{(i)}, \boldsymbol{x}_0^{(i)}) q(\boldsymbol{x}_{t-1}^{(i)} | \boldsymbol{x}_0^{(i)})}{q(\boldsymbol{x}_t^{(i)} | \boldsymbol{x}_0^{(i)})}\\
		&= \text{Cat} \left( \boldsymbol{x}_{t-1}^{(i)}; p = \frac{\boldsymbol{x}_t^{(i)} \boldsymbol{Q}_t^\top \odot \boldsymbol{x}_0^{(i)} \overline{\boldsymbol{Q}}_{t-1}}{\boldsymbol{x}_0^{(i)} \overline{\boldsymbol{Q}}_t \boldsymbol{x}_t^{(i)\top}} \right).
	\end{split} 
\end{equation}

\subsection{Drawbacks of the Factorization Approach}
\label{section: Drawbacks of Factorization}

Existing discrete diffusion models compute transition probability or ground-truth posterior in each dimension. This section quantitatively analyzes the impact of this factorization approach on fitting error. Intuitively, as the correlation among the dimensions of data increases, the performance of the discrete diffusion model decreases. 
Therefore, we consider a dataset \(\mathcal{D}_{\text{FC}}\), where each sample exhibits the strongest correlation among its dimensions, and derive the theoretical fitting errors for both existing and quantum models, thereby demonstrating the advantages of modeling joint distributions with the latter.
The fully correlated dataset \(\mathcal{D}_{\text{FC}}\) is mathematically defined in the following.
\begin{definition}
\label{Definition: Fully Correlated Dataset}
	(Fully Correlated Dataset $\mathcal{D}_{\text{FC}}$) Let \(\boldsymbol{x} = (x^{(1)}, x^{(2)}, \dots, x^{(N)})\) represent an \(N\)-dimensional discrete data point, where each component \(x^{(i)}\) can take \(K\) distinct values. The data exhibits full correlation across all dimensions if the following conditions hold:

	1. The reference dimension \(x^{(1)}\) takes exactly \(K\) distinct values, i.e., \(x^{(1)} \in \{0, 1, \dots, K-1\}\), and each value of \(x^{(1)}\) uniquely determines a corresponding \(\boldsymbol{x}\); and
	
	2. There exist bijective functions $f_2, f_3, \dots$, and $f_N$ such that:
	$ 	x^{(j)} = f_j(x^{(1)}) \quad \text{for } j = 2, \dots, N. $
	Thus, the dataset \(\mathcal{D}_{\text{FC}}\) is restricted to exactly \(K\) distinct samples:
	\begin{equation*}
		\begin{split}
\mathcal{D}_{\text{FC}} = \{&\boldsymbol{x} = (x^{(1)}, x^{(2)}, \dots, x^{(N)})\mid \\  &x^{(j)} = f_j(x^{(1)}),  x^{(1)} \in \{0, 1, \dots, K-1\}\}.
		\end{split}	
	\end{equation*}
	Each sample in \(\mathcal{D}_{\text{FC}}\) is associated with a unique value of \(x^{(1)}\), ensuring no two samples share the same value. 
	We assume that the dataset follows a uniform distribution, $ {P}(\boldsymbol{x}) = \frac{1}{K}, \boldsymbol{x} \in \mathcal{D}_{\text{FC}}$.
\end{definition}
We can derive the following theorem to reveal the results of comparing existing and quantum discrete diffusion models:
\begin{theorem}
  \label{thm:KL_bounds}  
  Let \(\boldsymbol{x} = (x^{(1)},\ldots,x^{(N)})\) be an \(N\)-dimensional discrete random variable drawn from \(\mathcal{D}_{\text{FC}}\) defined in \cref{Definition: Fully Correlated Dataset}, where each \(x^{(i)}\) takes \(K\) possible values, and let \(P(\boldsymbol{x})\) be its distribution.  
  Define the QD3PM joint distribution \( Q_{q}(\boldsymbol{x}) \). Then:
  \begin{equation}
  	\label{eq:kl with quantum}
    D_{\text{KL}}(P \parallel Q_q) = 0.
  \end{equation}
  For existing models with factorized distributions 
$
    Q_{c}(\boldsymbol{x}) := \prod_{i=1}^N Q_{c}(x^{(i)}), 
    \quad Q_{c}(x^{(i)}) = P(x^{(i)}),
$
  it holds that:
  \begin{equation}
  	\label{eq: kl with classical}
    D_{\text{KL}}(P \parallel Q_{c}) = (N-1)\log K.
  \end{equation}
\end{theorem}
\begin{proof}
	The proof is provided in \cref{section: proof of KL}. It begins by expressing the KL divergence in terms of entropy and mutual information, utilizing the definition of KL divergence and the factorized form of $ Q_c(\boldsymbol{x}) $. The fitting error of \( D_{\text{KL}}(P \parallel Q_{c}) \) is derived from the maximum marginal entropy, as \( P(x^{(i)}) \) is uniform, and the minimum joint entropy, as \( P(\boldsymbol{x}) \) is concentrated on \( K \) points, resulting in the error \( (N-1)\log K \).
\end{proof}

\cref{thm:KL_bounds} shows that QD3PM perfectly fits maximally correlated data by modeling the joint distribution, while factorized methods face a worst-case fitting error upper bound of $\mathcal{O}(N \log K)$. The actual fitting error usually remains below this bound. Stronger dimensional correlation pushes the fitting error nearer to the maximum error level of $(N-1)\log K$. This reveals a fundamental weakness of existing discrete diffusion models for high-dimensional correlated data, emphasizing the need for quantum computing to address this.

\subsection{Quantum Discrete Denoising Diffusion Probabilistic Model (QD3PM)}
\label{section: QD3PM}

\subsubsection{General Framework}

The framework of QD3PM is similar to its classical counterpart \cite{austin2021structured}, as shown in \cref{fig:framework}. At each timestep \( t \), the diffusion process encodes the classical data \( \boldsymbol{x}_{t-1} \) into a quantum state \( |\boldsymbol{x}_{t-1}\rangle \), uses a quantum channel to diffuse the joint distribution of all data dimensions, and measures the polluted state to obtain the next timestep's data \( \boldsymbol{x}_t \). The denoising process is implemented by a trainable quantum circuit, which predicts the denoised distribution based on noisy data \( \boldsymbol{x}_t \) and obtains cleaner data \( \boldsymbol{x}_{t-1} \) through measurements.
\begin{figure}[t]
	\centering
	\includegraphics[width=1\linewidth]{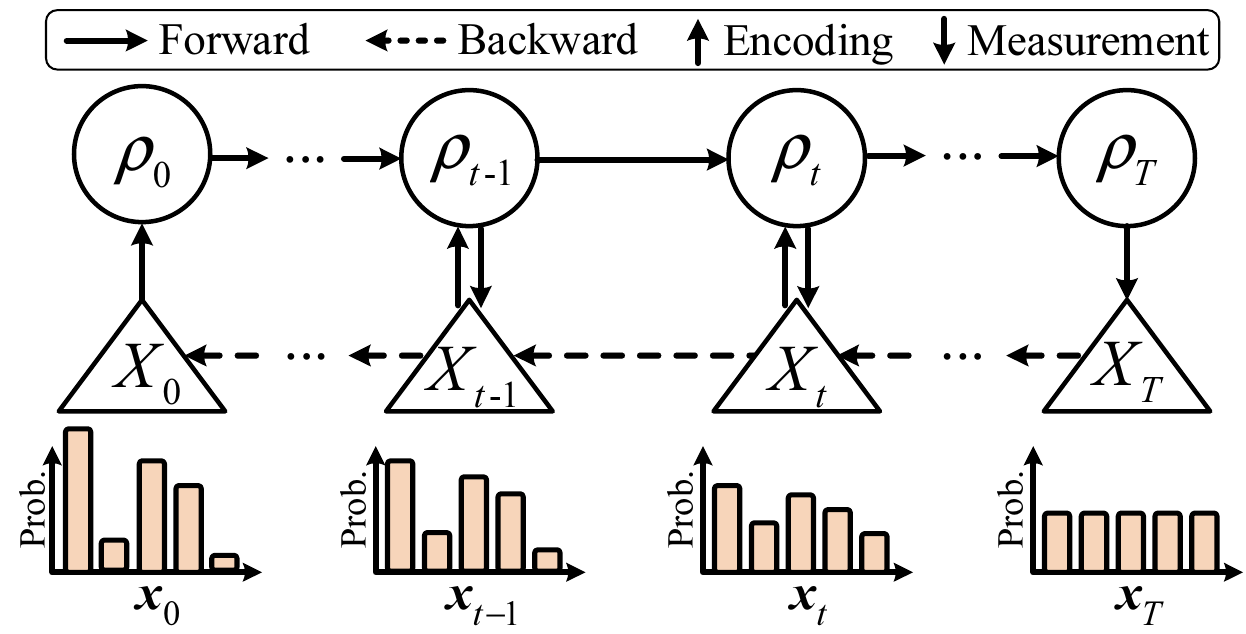}
	\caption{The framework of QD3PM.}
	\label{fig:framework}
\end{figure}

\subsubsection{Diffusion Forward Process}

In this paper, we consider generating binary data, where each dimension of a sample can only be 0 or 1. For problems with non-binary data, $ \lceil \log_2(K^N) \rceil $ qubits can be used to handle them. At timestep \( t \), an \( N \)-dimensional data \( \boldsymbol{x}_{t-1} = \{x_{t-1}^{(1)}, x_{t-1}^{(2)}, \ldots, x_{t-1}^{(N)}\} \)
%, where subscript $ t-1 $ denotes the timestep $ t-1 $ and the superscript $ (i) $ denotes the $ i $-th dimension of the data, 
is encoded by using basis encoding as a computational basis state in Hilbert space:
\begin{equation}
	\label{eq: encoding}
	\boldsymbol{x}_{t-1} \rightarrow |\boldsymbol{x}_{t-1}\rangle = |x_{t-1}^{(1)} x_{t-1}^{(2)} \cdots x_{t-1}^{(N)}\rangle.
\end{equation}
Applying a depolarizing channel \cite{nielsen2010quantum} to \( |\boldsymbol{x}_{t-1}\rangle \) yields the quantum state at the next timestep:
\begin{equation}
	\begin{split}
		\label{eq:one step depolarizing}
		\rho_t = &(1-\alpha_t)\boldsymbol{I}/d + \alpha_t|\boldsymbol{x}_{t-1}\rangle\langle \boldsymbol{x}_{t-1}|\\
		=&\sum_{\boldsymbol{x}} q(X_t=\boldsymbol{x} | \boldsymbol{x}_{t-1}) |\boldsymbol{x}\rangle\langle \boldsymbol{x}|,
	\end{split}
\end{equation}
where the scalar \( q(X_t = \boldsymbol{x} | \boldsymbol{x}_{t-1}) \) represents the probability of sampling outcome \( X_t = \boldsymbol{x} \), conditioned on the previous value \( \boldsymbol{x}_{t-1} \). The diffusion parameter \( \alpha_t \in [0, 1] \), \(\forall t \). $d=2^N$ is the dimension of Hilbert space and $ \boldsymbol{I} $ is the identity matrix. 
The state \( \rho_t\) encodes the joint distribution of all dimensions of \( \boldsymbol{x}_t \) conditioned on \( \boldsymbol{x}_{t-1} \).
Measuring \( \rho_t \) in the computational basis results in a sample from the probability distribution of \( \boldsymbol{x}_t \):
\begin{equation}
	\boldsymbol{x}_t \sim q(\boldsymbol{x}_t | \boldsymbol{x}_{t-1}) = \text{Tr}(|\boldsymbol{x}_t\rangle\langle \boldsymbol{x}_t| \rho_t).
\end{equation}
Where $ \text{Tr} $ is the trace operation \cite{nielsen2010quantum}.
The depolarizing channel decreases the probability of measuring the original \( |\boldsymbol{x}_{t-1}\rangle \) state and redistributes a uniform probability \( (1-\alpha_t)/d \) to other states.

Since \( \rho_t \) is diagonal, measuring it in the computational basis and re-encoding the result via (\ref{eq: encoding}) yield a quantum system with a density matrix equal to \( \rho_t \). Therefore, the relationship between \( \rho_t \) and \( |\boldsymbol{x}_0\rangle \) can be expressed as \cite{chen2024quantum}:
\begin{equation}
	\begin{split} 
		\rho_t = &(1 - \bar{\alpha}_t)\boldsymbol{I}/d + \bar{\alpha}_t|\boldsymbol{x}_0\rangle\langle \boldsymbol{x}_0|\\
		=&\sum_{\boldsymbol{x}} q(X_t=\boldsymbol{x} | \boldsymbol{x}_{0}) |\boldsymbol{x}\rangle\langle \boldsymbol{x}|,
	\end{split}
\end{equation}
where \( \bar{\alpha}_t = \prod_{i=1}^{t} \alpha_i \). The transition probability between \( \boldsymbol{x}_t \) at any timestep \( t \) and \( \boldsymbol{x}_0 \) is given as:
\begin{equation} 
	\label{eq:sample t step}
	\begin{split} 
		q(\boldsymbol{x}_t | \boldsymbol{x}_0)
		= &\text{Tr}\left[|\boldsymbol{x}_t\rangle\langle \boldsymbol{x}_t| \left( (1 - \bar{\alpha}_t)\boldsymbol{I}/d + \bar{\alpha}_t|\boldsymbol{x}_0\rangle\langle \boldsymbol{x}_0| \right) \right].
	\end{split} 
\end{equation} 
\( p(\boldsymbol{x}_T) \) approaches a uniform distribution as \( \bar{\alpha}_T \approx 0 \).
Therefore, we can obtain noisy sample \( \boldsymbol{x}_t \) from \( \boldsymbol{x}_0 \) via (\ref{eq:sample t step}).

In this paper, we use the cosine noise schedule \cite{nichol2021improved} to select \( \alpha_t \):
$
	\alpha_t = \frac{\bar{\alpha}_t}{\bar{\alpha}_{t-1}},
$
where
$ 	\bar{\alpha}_t = \frac{g(t)}{g(0)}, \quad g(t) = \cos \left(\frac{t / T + s}{1 + s} \cdot \frac{\pi}{2}\right)^2, $
with \( s \) being a small offset hyperparameter.

\subsubsection{Denoising Backward Process}
In this section, the denoising process \( p_{\boldsymbol{\theta}}(\boldsymbol{x}_{t-1} | \boldsymbol{x}_t) \) is directly implemented with quantum circuits. 
In \cref{section: sample in one step}, we introduce how to train the quantum circuit to model $p_{\boldsymbol{\theta}}(\boldsymbol{x}_{0} | \boldsymbol{x}_t)$.
Two key factors are considered in designing such a circuit. First, since \( p(\boldsymbol{x}_T) \) is a uniform distribution, its corresponding density matrix \( \rho_T \) is a completely mixed state. When inputting \( \boldsymbol{x}_T \) into the quantum circuit with basis encoding, a general unitary circuit does not change its distribution since the completely mixed state \( \rho_T \) remains invariant under any unitary transformation $U$, as $U\rho_TU^{\dagger}=\rho_T$. Therefore, either designing non-unitary circuits or avoiding basis encoding is necessary. Second, to minimize the total number of QD3PM model parameters, temporal information $t$ is used as input, enabling weight sharing across the denoising circuits at each step, thus avoiding the need for a distinct set of parameters for each of the $T$ denoising iterations.

\begin{figure}[t]
	\centering
	\includegraphics[width=1\linewidth]{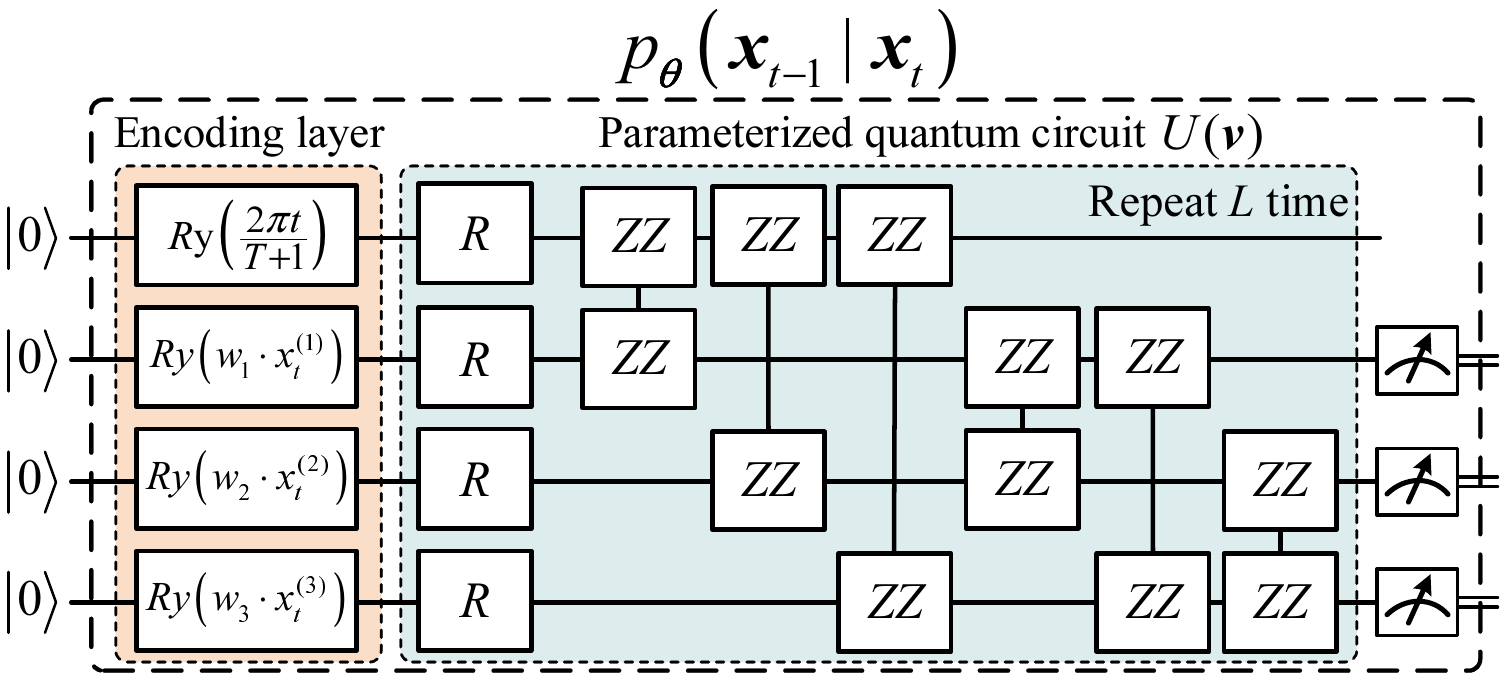}
	\caption{The denoising circuit architecture for noisy data $\boldsymbol{x}_t=\{x_t^{1},x_t^{2},x_t^{3}\}$. We can sample a denoising data $\boldsymbol{x}_{t-1}$ from this circuit. The rotation gate $ R $ is defined as \( R = R_xR_yR_x \). The gate ZZ represents the two-qubit gate $ZZ(v)=\exp{(-iv (Z\otimes Z) /2)}$. The trainable parameters in denoising circuit are $\boldsymbol{\theta}=\{\boldsymbol{w},\boldsymbol{v}\}$.}
	\label{fig:PQC}
\end{figure}

With these considerations, we design the circuit in \cref{fig:PQC} to denoise noisy data \( \boldsymbol{x}_t = \{x_t^{(1)}, x_t^{(2)}, \dots, x_t^{(N)}\} \). The first layer of the circuit encodes \( t \) and \( \boldsymbol{x}_t \), resulting in the quantum state:
\begin{equation}
	|\psi(t, \boldsymbol{x}_t)\rangle = R_y\left(\frac{2\pi \cdot t}{T+1}\right)|0\rangle \otimes_{i=1}^{N} R_y\left(w_i \cdot x_t^{(i)}\right)|0\rangle, 
\end{equation}
where $R_y$ denotes the $y$-axis rotation gate commonly used in quantum computing \cite{nielsen2010quantum} and \( \boldsymbol{w} = \{w_i\}_{i=1}^N \) represents trainable parameters. After the encoding layer, a parameterized quantum circuit (PQC) is applied to obtain the output state:
\begin{equation}
	|\psi_{\text{out}}\rangle = U(\boldsymbol{v})|\psi(t, \boldsymbol{x}_t)\rangle.
\end{equation}
The PQC $U(\boldsymbol{v})$ consists of \( L \) layers with the same structure, but each layer has different parameters. $\boldsymbol{v}$ stands for trainable parameters in PQC. A measurement in the computational basis is performed on the qubits, from the second to the last, for a total of \( N \) qubits. The reduced density matrix of the last $ N $ qubits is:
\begin{equation}
	\begin{split}
		\rho_{\text{out}} = &\text{Tr}_1(|\psi_{\text{out}}\rangle\langle\psi_{\text{out}}|)
		\\
		=&\sum_{\boldsymbol{x}} p_{\boldsymbol{\theta}}(X_{t-1}=\boldsymbol{x} | \boldsymbol{x}_{t}) |\boldsymbol{x}\rangle\langle \boldsymbol{x}|,
	\end{split}
\end{equation}
where the scalar \( p_{\boldsymbol{\theta}}(X_{t-1}=\boldsymbol{x} | \boldsymbol{x}_{t}) \) represents the probability of sampling outcome \( X_{t-1} = \boldsymbol{x} \), conditioned on the noise data \( \boldsymbol{x}_{t} \). $\text{Tr}_1$ represents the partial trace operation that traces out the first qubit system. In general, \( \rho_{\text{out}} \) is a mixed state encoding the model-predicted distribution at timestep \( t-1 \) given \( \boldsymbol{x}_t \). 
We can measure it to generate a denoised sample:
\begin{equation}
	\boldsymbol{x}_{t-1} \sim p_{\boldsymbol{\theta}}(\boldsymbol{x}_{t-1} | \boldsymbol{x}_t) = \text{Tr}\left(|\boldsymbol{x}_{t-1}\rangle\langle\boldsymbol{x}_{t-1}| \rho_{\text{out}} \right),
\end{equation}
where \( \boldsymbol{\theta}=\{\boldsymbol{w}, \boldsymbol{v}\} \) represents the trainable parameters in the denoising process $ p_{\boldsymbol{\theta}}(\boldsymbol{x}_{t-1} | \boldsymbol{x}_t) $.

\subsubsection{Derivation of Posterior Distribution State}
To train the denoising process, a ground-truth posterior distribution is required. In existing discrete diffusion models, this posterior is derived via Bayes' theorem. However, this method cannot be directly applied to QD3PM, as it models the joint distribution of the data and performs diffusion at the quantum level. Instead, we use the quantum conditional state formalism and quantum Bayes' theorem \cite{leifer2013towards} to derive the form of the posterior distribution for training. \cref{section:Quantum Conditional States Formalism} provides a brief review of the formalism and quantum Bayes' theorem.
For clarity in the following derivations, we outline some key notation conventions:
\begin{remark}
	\label{remark1} In conventional quantum notation, a quantum state \( \rho_{t-1} \) evolves to \( \rho_t \) through channel \( \mathcal{E} \). Both \( \rho_{t-1} \) and \( \rho_t \) reside in the same quantum register and Hilbert space, with the random variable resulting from measuring \( \rho_t \) denoted as \( X_t \), and its values represented by \( \boldsymbol{x}_t \). 
 	However, the conditional quantum state formalism distinguishes Hilbert spaces based on temporal and spatial factors. Though \( \rho_{t-1} \) and \( \rho_t \) reside on the same register, the formalism treats them as belonging to distinct Hilbert spaces, \( \mathcal{H}_{t-1} \) and \( \mathcal{H}_t \), due to their different timesteps. Thus, channel \( \mathcal{E} \) is redefined as \( \mathcal{E}_{\mathcal{H}_t | \mathcal{H}_{t-1}} \), which maps operators from \( \mathcal{H}_{t-1} \) to \( \mathcal{H}_t \), reflecting the evolution from timestep \( t-1 \) to \( t \).
\end{remark}

% We now turn to the problem of deriving the ground-truth posterior distribution \( q(\boldsymbol{x}_{t-1} | \boldsymbol{x}_t, \boldsymbol{x}_0) \) in the quantum scenario, where \( X_0=\boldsymbol{x}_0 \) and \( X_t=\boldsymbol{x}_t \) are known. 
% We employ quantum Bayes' theorem to obtain a posterior distribution state according to the known conditions, whose measurement in the computational basis gives the required posterior distribution.

We now turn to the problem of deriving the ground-truth posterior distribution \( q(\boldsymbol{x}_{t-1} | \boldsymbol{x}_t, \boldsymbol{x}_0) \) in the quantum scenario, given known conditions $X_0=\boldsymbol{x}_0$ and $X_t=\boldsymbol{x}_t$. This distribution is obtained by measuring the posterior quantum state $\tilde{\rho}_{t-1}$ in the computational basis. To derive $\tilde{\rho}_{t-1}$, we leverage the relationship among $\boldsymbol{x}_0$, $\rho_t$, $\rho_{t-1}$, and $\boldsymbol{x}_t$, as illustrated in \cref{fig:Inference_rho_t_1}. The process begins by formulating a direct relationship between space $\mathcal{H}_{t-1}$ and the random variable $X_t$. This is achieved by marginalizing the intermediate space $\mathcal{H}_t$ using the quantum conditional state formalism. Subsequently, the posterior state $\tilde{\rho}_{t-1}$ is computed by applying the quantum Bayesian update rule \cite{leifer2013towards}, conditioned on the observed measurement $X_t=\boldsymbol{x}_t$.

\begin{figure}[t]
	\centering
	\includegraphics[width=0.9\linewidth]{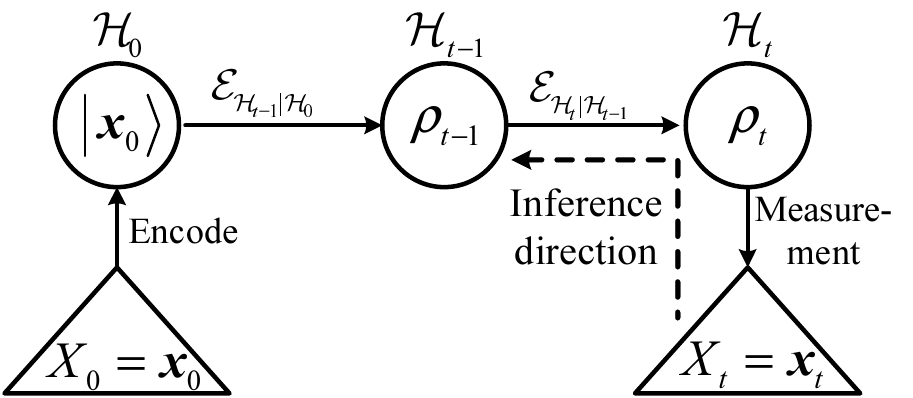}
	\caption{The causal relationship in deriving the posterior state.}
	\label{fig:Inference_rho_t_1}
\end{figure}

% The relationship among \( \boldsymbol{x}_0 \), \( \rho_t \), \( \rho_{t-1} \), and \( \boldsymbol{x}_t \) is illustrated in \cref{fig:Inference_rho_t_1}. 
% The rationale for deriving the updated state $\tilde{\rho}{t-1}$ is as follows: leveraging the dependencies depicted in \cref{fig:Inference_rho_t_1} and established conditions, a direct relationship between $\rho{t-1}$ and the measurement $X_t$ is first formulated. This requires marginalizing the intermediate state $\rho_t$ through the quantum conditional state formalism. The posterior state $\tilde{\rho}_{t-1}$ is then computed using the quantum Bayesian update rule (\cref{eq: quantum bayes updated}), conditioned on the observed outcome $X_t=\boldsymbol{x}_t$.

First, the transformation of \( \rho_{t-1} \) to \( \rho_t \) through channel \( \mathcal{E}_{\mathcal{H}_t | \mathcal{H}_{t-1}} \) is expressed via the Choi-Jamio\l{}kowski isomorphic \cite{jamiolkowski1972linear,choi1975completely,leifer2013towards} operator, i.e.,
\begin{equation}
	\label{equ:J-isomorphic}
	\varrho_{\mathcal{H}_t | \mathcal{H}_{t-1}} = \left( \mathcal{E}_{\mathcal{H}_t | \mathcal{H}^{'}_{t-1}} \otimes \boldsymbol{I}_{\mathcal{H}_{t-1}} \right) \left( \left| \Phi^+ \right\rangle \left\langle \Phi^+ \right|_{\mathcal{H}_{t-1} \mathcal{H}_{t-1}'} \right),
\end{equation}
where \( \mathcal{H}_{t-1}' \) is a copy of \( \mathcal{H}_{t-1} \), and \( |\Phi^+\rangle = \frac{1}{\sqrt{d}} \sum_{\boldsymbol{i}=0}^{d-1} |\boldsymbol{i}\rangle \otimes |\boldsymbol{i}\rangle \) is a maximally entangled state with \( d = 2^N \) being the dimension of the system.  $ \varrho_{\mathcal{H}_t|\mathcal{H}_{t-1}} $ is an operator on the composite space $ \mathcal{H}_t \otimes \mathcal{H}_{t-1} $. The process of measuring \( \rho_t \) in the computational basis to obtain random variable \( X_t \) is represented via a positive operator:
\begin{equation}
	\label{equ:hybrid operator}
	\varrho_{X_t | \mathcal{H}_t} = \sum_{\boldsymbol{x}} |\boldsymbol{x}\rangle \langle \boldsymbol{x}|_{X_t} \otimes |\boldsymbol{x}\rangle \langle \boldsymbol{x}|_{\mathcal{H}_t}.
\end{equation}
Therefore, the direct relationship between \( X_t \) and \( \rho_{t-1} \) is described via a causal conditional state \( \varrho_{X_t | \mathcal{H}_{t-1}} \), i.e.,
\begin{equation}
	\label{equ:partial trace}
	\varrho_{X_t | \mathcal{H}_{t-1}} = \text{Tr}_{\mathcal{H}_t} \left( \varrho_{X_t | \mathcal{H}_t} \cdot \varrho_{\mathcal{H}_t | \mathcal{H}_{t-1}} \right).
\end{equation}
Finally, when the measurement outcome \( X_t = \boldsymbol{x}_t \) is obtained, \( \rho_{t-1} \) can be updated according to the quantum Bayes' theorem:
\begin{equation}
	\begin{split} 
		\label{equ:update rule}
		\rho_{t-1} &\rightarrow \tilde{\rho}_{{t-1} | X_t = \boldsymbol{x}_t,X_0 = \boldsymbol{x}_0} \\
		&= \varrho_{X_t = \boldsymbol{x}_t | \mathcal{H}_{t-1}} \star \left( \rho_{t-1 | X_0 = \boldsymbol{x}_0} \rho_{X_t = \boldsymbol{x}_t}^{-1} \right),
	\end{split} 
\end{equation}
where \( \rho_{X_t = \boldsymbol{x}_t} \) is the probability of obtaining the measurement result \( X_t = \boldsymbol{x}_t \) given \( X_0 = \boldsymbol{x}_0 \). The $\star$ product is defined as $M \star N = N^{\frac{1}{2}} M N^{\frac{1}{2}}$. We derive the update rule for \( \rho_{t-1} \) under any non-unitary diffusion channel \( \mathcal{E}_{\mathcal{H}_t | \mathcal{H}_{t-1}} \), which is not applicable to unitary channels and is therefore excluded.
The ground-truth posterior distribution is obtained by measuring \( \tilde{\rho}_{t-1} \) in the computational basis:
\begin{equation}
	q(\boldsymbol{x}_{t-1} | \boldsymbol{x}_t, \boldsymbol{x}_0) = \text{Tr} \left( |\boldsymbol{x}_{t-1}\rangle \langle \boldsymbol{x}_{t-1}| \tilde{\rho}_{t-1} \right).
\end{equation}
For our implementation, \( \mathcal{E}_{\mathcal{H}_t | \mathcal{H}_{t-1}} \) is a depolarizing channel parameterized by \( \alpha_t \), leading to the following theorem:
\begin{theorem}
	\label{thm:non-unitary posterior}
	If a diffusion process is a depolarizing channel parameterized by \(\alpha_t\) in (\ref{eq:one step depolarizing}), and all density matrices during diffusion are diagonal, the posterior state used for training the denoising model, representing the ground-truth posterior distribution, can be derived from the quantum conditional state formalism and quantum Bayes' theorem as follows:
	\begin{equation}
		\begin{split} 
			\label{equ:non-unitary posterior}
			&\tilde{\rho}_{t-1 | X_t = \boldsymbol{x}_t, X_0 = \boldsymbol{x}_0} \\
			= &\frac{\alpha_t \langle \boldsymbol{x}_t | \rho_{t-1} | \boldsymbol{x}_t \rangle |\boldsymbol{x}_t\rangle \langle \boldsymbol{x}_t| + \frac{1}{d^2}(1-\alpha_t) \rho_{t-1}}{\alpha_t \langle \boldsymbol{x}_t | \rho_{t-1} | \boldsymbol{x}_t \rangle + \frac{1}{d^2} (1-\alpha_t)},
		\end{split} 
	\end{equation}
	where \( \rho_{t-1} = (1 - \bar{\alpha}_{t-1}) \boldsymbol{I}/d + \bar{\alpha}_{t-1} |\boldsymbol{x}_0\rangle \langle \boldsymbol{x}_0| \).
\end{theorem}
\begin{proof} 
	The proof is based on the fact that \( \rho_t \) is diagonal at each timestep, and proceeds through (\ref{equ:J-isomorphic})-(\ref{equ:update rule}) as detailed in \cref{section: proof of posterior}.
\end{proof}

\begin{figure}[t]
	\centering
	\includegraphics[width=1\linewidth]{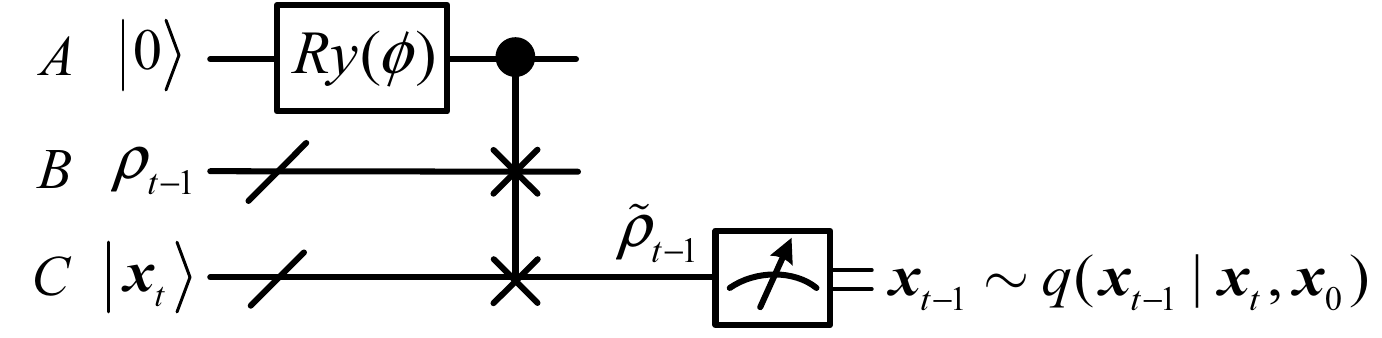}
	\caption{Circuit architecture to generate the posterior state.}
	\label{fig:pos_circuit1}
\end{figure}
To compute the posterior state in (\ref{equ:non-unitary posterior}) using quantum circuits, we express it as:
\begin{equation}
	\tilde{\rho}_{t-1|X_t=\boldsymbol{x}_t,X_0=\boldsymbol{x}_0} = \gamma_1 |\boldsymbol{x}_t\rangle\langle \boldsymbol{x}_t| + \gamma_2 \rho_{t-1},
\end{equation}
where \( \gamma_1 + \gamma_2 = 1 \), \( \gamma_1 = \frac{ \alpha_t \langle \boldsymbol{x}_t | \rho_{t-1} | \boldsymbol{x}_t \rangle }{ \alpha_t \langle \boldsymbol{x}_t | \rho_{t-1} | \boldsymbol{x}_t \rangle +  \frac{1}{d^2} (1-\alpha_t) } \), and \( \gamma_2 = \frac{ \frac{1}{d^2} (1-\alpha_t) }{ \alpha_t \langle \boldsymbol{x}_t | \rho_{t-1} | \boldsymbol{x}_t \rangle +  \frac{1}{d^2} (1-\alpha_t) } \). We design a quantum circuit as shown in \cref{fig:pos_circuit1}.
To compute \( \tilde{\rho}_{t-1} \), we set the rotation parameter of the \( R_y \) gate as:
\begin{equation}
	\phi = 2 \arccos \left( \sqrt{ \frac{ \alpha_t \langle \boldsymbol{x}_t | \rho_{t-1} | \boldsymbol{x}_t \rangle }{ \alpha_t \langle \boldsymbol{x}_t | \rho_{t-1} | \boldsymbol{x}_t \rangle + (1 - \alpha_t) \frac{1}{d^2} } } \right).
\end{equation}
This setup allows the computation of \( \tilde{\rho}_{t-1} \). Measuring in computational basis gives a distribution that matches the ground-truth posterior \( q(\boldsymbol{x}_{t-1}|\boldsymbol{x}_t, \boldsymbol{x}_0) \). The value of \( \langle \boldsymbol{x}_t | \rho_{t-1} | \boldsymbol{x}_t \rangle \) can be determined as:
\begin{equation}
	\langle \boldsymbol{x}_t | \rho_{t-1} | \boldsymbol{x}_t \rangle = 
	\begin{cases} 
		\bar{\alpha}_{t-1} + \frac{1 - \bar{\alpha}_{t-1}}{d} , &  \boldsymbol{x}_t = \boldsymbol{x}_0, \\
		\frac{1 - \bar{\alpha}_{t-1}}{d} , & \boldsymbol{x}_t \neq \boldsymbol{x}_0. 
	\end{cases} 
\end{equation}
\cref{section: pose circuit derivation} provides the detailed computation of the circuit.

\subsubsection{Training and Generation}
The optimization objective of QD3PM aligns with that of existing discrete diffusion models, both aiming to optimize the variational lower bound $L_{vb}$ as shown in (\ref{eq:lower bound}). Since the diffusion process has no trainable parameters, the optimization focuses on the $\mathcal{L}_{t-1}$ and $\mathcal{L}_0$ terms. Thus, the objective function of QD3PM is:
\begin{equation}
	\begin{split}
		\label{eq: loss function}
		\mathcal{L} = & \mathcal{L}_{t-1} + \mathcal{L}_{0}\\  
		=&\mathbb{E}_{q(\boldsymbol{x}_0),\ t \sim \mathcal{U}(2, T),\ q(\boldsymbol{x}_t \mid \boldsymbol{x}_0)} [ \\ 
		&\quad D_{\text{KL}}\left(q(\boldsymbol{x}_{t-1} \mid \boldsymbol{x}_t, \boldsymbol{x}_0) \,\|\, p_{\boldsymbol{\theta}}(\boldsymbol{x}_{t-1} \mid \boldsymbol{x}_t) \right) ]\\
		&-\mathbb{E}_{q(\boldsymbol{x}_0),q(\boldsymbol{x}_1 | \boldsymbol{x}_0)}  [\log p_{\boldsymbol{\theta}}(\boldsymbol{x}_0 | \boldsymbol{x}_1)],
	\end{split}
\end{equation}
where $\mathcal{U}(2, T)$ is a uniform distribution between 2 and $T$.

QD3PM is an implicit quantum generative model, meaning that it generates samples but does not provide a full probability mass function in polynomial time. For implicit generative models, training with Maximum Mean Discrepancy (MMD) loss can mitigate the barren plateau problem according to \cite{rudolph2024trainability}. Thus, we minimize the KL term in (\ref{eq: loss function}) by minimizing the following MMD loss:
\begin{equation}
	\begin{split}
\mathcal{L}_{\mathrm{MMD}}(p_{\boldsymbol{\theta}},q)&=  \mathbb{E}_{\boldsymbol{x}, \boldsymbol{y} \sim p_{\boldsymbol{\theta}}}[\mathcal{K}(\boldsymbol{x}, \boldsymbol{y})]+\mathbb{E}_{\boldsymbol{x}, \boldsymbol{y} \sim q}[\mathcal{K}(\boldsymbol{x}, \boldsymbol{y})]\\&-2 \mathbb{E}_{\boldsymbol{x} \sim p_{\boldsymbol{\theta}}, \boldsymbol{y} \sim q}[\mathcal{K}(\boldsymbol{x}, \boldsymbol{y})],
	\end{split}
\end{equation}
where bandwidth $\sigma$ is a hyperparameter that must be carefully chosen, as it impacts both the model's trainability and ability to fit the target distribution. In Appendix \ref{section: Bandwidth}, we demonstrate how to set the bandwidth.

After training, we can generate clean samples $\boldsymbol{x}_0$ by iteratively sampling for $T$ steps with $p_{\boldsymbol{\theta}}(\boldsymbol{x}_{t-1} | \boldsymbol{x}_t)$, starting from timestep $T$.
The training and generation algorithms are shown in \cref{section: training and generation algorithms}.

\subsection{Learning to Sample in One Step}
\label{section: sample in one step}
In this section, we present a modified version of QD3PM to enable the circuit to learn $p_{\boldsymbol{\theta}}(\boldsymbol{x}_{0}|\boldsymbol{x}_{t})$, thereby achieving single-step sampling after training.
While one-shot generation has been explored for classical diffusion models \cite{yin2024one,zhou2024score}, our work, to the best of our knowledge, is the first to achieve this within a fully quantum diffusion framework.
We design our quantum circuits with the objective of learning $p_{\boldsymbol{\theta}}(\boldsymbol{x}_{0}|\boldsymbol{x}_{t})$, which is then used to construct $p_{\boldsymbol{\theta}}(\boldsymbol{x}_{t-1}|\boldsymbol{x}_t)$. This approach maintains the same diffusion and training processes as in the previous section.

\textbf{Principle.} The transition probability $p_{\boldsymbol{\theta}}(\boldsymbol{x}_{t-1}|\boldsymbol{x}_t)$ can be expressed by marginalizing over the original data $\tilde{\boldsymbol{x}}_0$:
\begin{equation}
\label{eq: distribution get xt_1}
    p_{\boldsymbol{\theta}}(\boldsymbol{x}_{t-1}|\boldsymbol{x}_t)=\sum_{\tilde{\boldsymbol{x}}_0} p_{\boldsymbol{\theta}}(\tilde{\boldsymbol{x}}_{0}|\boldsymbol{x}_t) q(\boldsymbol{x}_{t-1}|\boldsymbol{x}_t,\tilde{\boldsymbol{x}}_0),
\end{equation}
where \( p_{\boldsymbol{\theta}}(\tilde{\boldsymbol{x}}_{0}|\boldsymbol{x}_t) \) is the probability distribution learned by the circuit. \( q(\boldsymbol{x}_{t-1}|\boldsymbol{x}_t,\tilde{\boldsymbol{x}}_0) \) is the measurement distribution over the posterior state deriving from (\ref{equ:non-unitary posterior}). 
Measuring the circuit produces an outcome \( 
\tilde{\boldsymbol{x}}_{0} \) with probability \( p_{\boldsymbol{\theta}}(\tilde{\boldsymbol{x}}_{0} | \boldsymbol{x}_{t}) \). Using basis encoding to represent it as state \( |\tilde{\boldsymbol{x}}_{0}\rangle \), the encoding register's density matrix can be represented as:
\begin{equation}
\label{eq: encode x0}
    \rho_{\text{enc}} = \sum_{\tilde{\boldsymbol{x}}_{0}} p_{\boldsymbol{\theta}}(\tilde{\boldsymbol{x}}_{0} | \boldsymbol{x}_{t}) \, |\tilde{\boldsymbol{x}}_{0}\rangle \langle \tilde{\boldsymbol{x}}_{0}|.
\end{equation}
% \cref{eq: encode x0} shows how the probability distribution $p_{\boldsymbol{\theta}}(\tilde{\boldsymbol{x}}_{0} | \boldsymbol{x}_{t})$ can be represented as a density matrix (a statistical mixture) over the basis states \( |\tilde{\boldsymbol{x}}_{0}\rangle \).
Next, we compute the posterior state using \cref{equ:non-unitary posterior}. We first apply a depolarizing channel to $\rho_{\text{enc}}$, yielding
\begin{equation}
\label{eq: encoding state after depolarizing}
    \rho_{\text{enc}}' = \sum_{\tilde{\boldsymbol{x}}_{0}} p_{\boldsymbol{\theta}}(\tilde{\boldsymbol{x}}_{0} | \boldsymbol{x}_{t}) \, \rho_{t-1},
\end{equation}
where
$
    \rho_{t-1} = (1 - \bar{\alpha}_{t-1}) \frac{\boldsymbol{I}}{d} + \bar{\alpha}_{t-1} |\tilde{\boldsymbol{x}}_0\rangle \langle \tilde{\boldsymbol{x}}_0|.
$
\cref{eq: encoding state after depolarizing} is correct since the channel is linear and preserves the convexity of density matrices \cite{wilde2013quantum}.
Then we get the before-measurement state according to the \cref{thm:non-unitary posterior}:
\begin{equation}
\label{eq: rho_final}
    \rho_{\text{final}} = \sum_{\tilde{\boldsymbol{x}}_{0}} p_{\boldsymbol{\theta}}(\tilde{\boldsymbol{x}}_{0} | \boldsymbol{x}_{t}) \, \tilde{\rho}_{t-1 | X_t = \boldsymbol{x}_t, X_0 = \tilde{\boldsymbol{x}}_0}.
\end{equation}
Measuring $\rho_{\text{final}}$ in the computational basis produces distribution $p_{\boldsymbol{\theta}}(\boldsymbol{x}_{t-1}|\boldsymbol{x}_t)$:
\begin{equation}% 
\begin{split}
\label{eq: marginalizaing x0}
&\text{Tr}(|\boldsymbol{x}_{t-1}\rangle \langle \boldsymbol{x}_{t-1}|\rho_{\text{final}} )\\= &\sum_{\tilde{\boldsymbol{x}}_{0}} p_{\boldsymbol{\theta}}(\tilde{\boldsymbol{x}}_{0} | \boldsymbol{x}_{t})\text{Tr}(|\boldsymbol{x}_{t-1}\rangle \langle \boldsymbol{x}_{t-1}|\tilde{\rho}_{t-1 | X_t = \boldsymbol{x}_t, X_0 = \tilde{\boldsymbol{x}}_0}) \\
=&\sum_{\tilde{\boldsymbol{x}}_{0}} p_{\boldsymbol{\theta}}(\tilde{\boldsymbol{x}}_{0} | \boldsymbol{x}_{t})q(\boldsymbol{x}_{t-1} | \boldsymbol{x}_{t}, \tilde{\boldsymbol{x}}_{0}).
\end{split}
\end{equation}

\begin{figure*}[t]
	\centering
	\includegraphics[width=\linewidth]{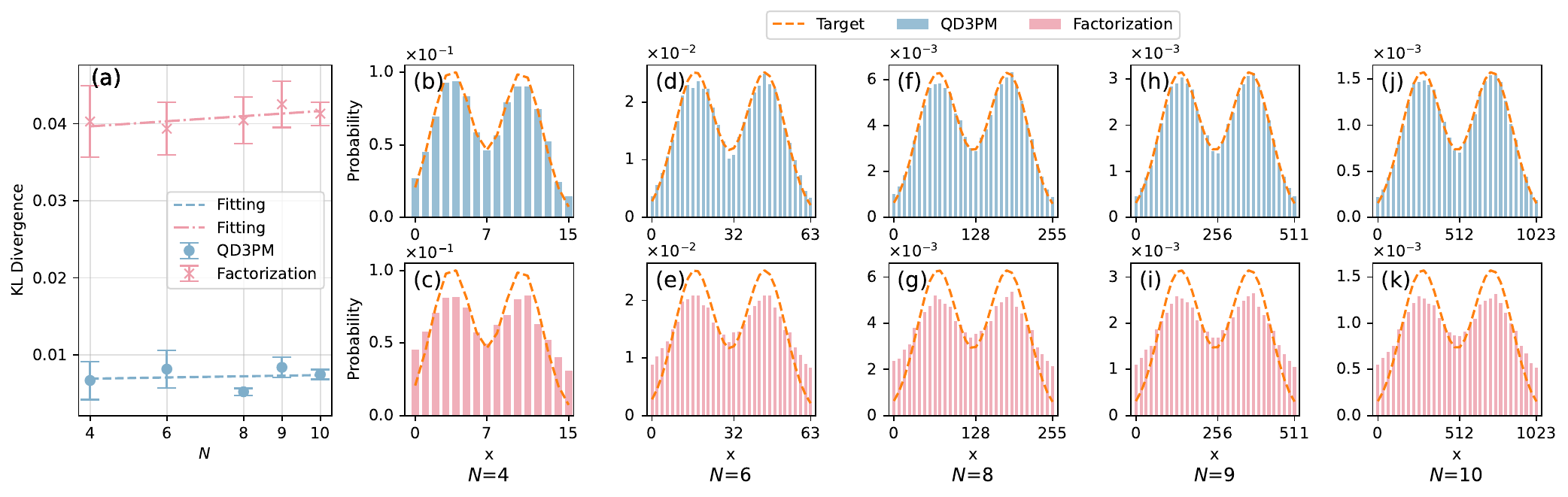}
	\caption{Numerical results of QD3PM and its factorized version fitting mixed Gaussian dataset. (a) The trend in KL divergence between generated and target distributions for different models as $N$ increases. (b)-(k) Visual comparison of fitting performance between QD3PM and the factorization method for various data sizes $N$.}
	\label{fig:8gaussian_kl_sample}
\end{figure*}

\textbf{Simplified calculation.}
The final sampling of $p_{\boldsymbol{\theta}}(\boldsymbol{x}_{t-1}|\boldsymbol{x}_t)$ depends  only on the diagonal elements of  $\rho_{\text{final}}$. Recognizing this, we can simplify the procedure by omitting the intermediate measurement of the circuit's output $\rho_{\text{out}}$. Instead, we leverage the linearity of the subsequent quantum operations (the depolarizing channel and posterior update derived from \cref{equ:non-unitary posterior}) by applying them directly to $\rho_{\text{out}}$. While the density matrix resulting from this direct application may differ from the theoretical $\rho_{\text{final}}$ (in Equation (\ref{eq: rho_final})) in its off-diagonal elements, the crucial main diagonal elements required for sampling $p_{\boldsymbol{\theta}}(\boldsymbol{x}_{t-1}|\boldsymbol{x}_t)$ are correctly reproduced. This streamlined approach elegantly achieves the correct outcome by operating directly on the circuit's output state, avoiding the cost of intermediate measurement and re-preparation.

\subsection{Numerical Simulations}
\label{section: simulation}
We evaluate the performance of our model on two datasets: Bars-and-Stripes (BAS) and the Mixed Gaussian ones. We assume that all the denoising circuits have an all-to-all qubit topology. The simulation setup is shown in \cref{section: simulations setup}. For details on loss metrics during training, and KL divergence to target distribution during generation, please refer to Appendix \ref{section: more result}.

\begin{figure*}[t]
	\centering
\includegraphics[width=1\linewidth]{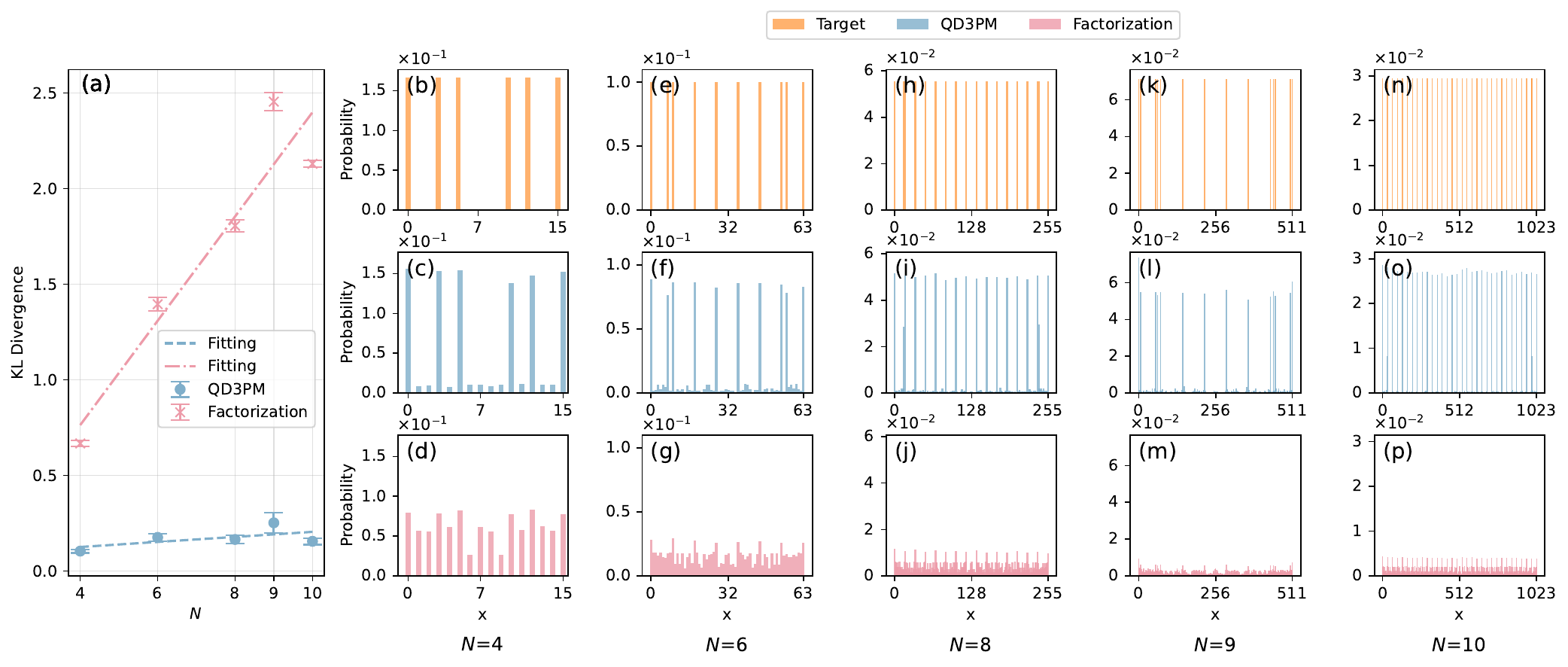}
\caption{Numerical results of QD3PM and its factorized version fitting BAS dataset. (a) The trend in KL divergence between generated and target distributions for different models as $N$ increases. (b)-(p) Visualizing Different Distributions:
The top row displays the BAS target distribution.
The middle row presents the QD3PM-generated distribution.
The bottom row shows the factorization method's generated distribution.}
    \label{fig:N_KL_compare}
\end{figure*}

\subsubsection{Training Circuit as $p_{\boldsymbol{\theta}}(\boldsymbol{x}_{t-1}|\boldsymbol{x}_t)$}

\textbf{Mixed Gaussian Dataset.}
% \cref{fig:8gaussian_kl_sample}(a) displays the results of QD3PM fitting a mixture of Gaussian distributions with different quantum circuit layers i.e., $ L = 2, 4, 6 $, and $8$. \cref{fig:8gaussian_kl_sample}(a) shows that the KL divergence between the model and target distribution decreases rapidly during training, indicating great trainability. The shaded region represents the standard deviation, which diminishes with increasing \( L \), highlighting the model’s improved stability with deeper circuits.
% \cref{fig:8gaussian_kl_sample}(b) displays the distribution of \( 10^6 \) samples generated by the trained model with \( L = 8 \), with the dashed line representing the target distribution. The generated samples closely match the target distribution, demonstrating the model's effective learning and fitting capabilities, even when starting from a uniform distribution. \cref{fig:8gaussian_kl_sample}(c) displays the distribution generated by factorization-based model. These results confirm the strong generative performance of QD3PM.
\cref{fig:8gaussian_kl_sample} showcases the performance of QD3PM in modeling mixed Gaussian distributions for systems with $N = {4, 6, 8, 9}$ and 10 qubits, contrasting its results with those from a factorization-based method. \cref{fig:8gaussian_kl_sample}(a) plots the KL divergence from the target distribution for both QD3PM and the factorization method as a function of the qubit count $N$. Error bars represent standard deviations, and dashed lines depict the fitted line. To more directly visualize the fitting performance of the two models, subfigures (b)-(k) show the sampling distributions of QD3PM and the factorization method after training for different $N$.

\cref{fig:8gaussian_kl_sample}(a) demonstrates that QD3PM consistently fits the target distribution well across variant $N$, maintaining an average KL divergence below 0.01 with no significant increasing trend as $N$ grows. This indicates that by modeling the joint distribution, QD3PM effectively captures the structural information of the data at different scales, thus achieving excellent fitting results. In contrast, factorization-based methods exhibit KL divergences around or above 0.04, with a slight tendency to increase with $ N $. This is because these methods disrupt the inherent correlations (or dependencies) within the data.
%	; as revealed by our \cref{thm:KL_bounds}, the KL fitting error tends to increase with N.}

Figures \ref{fig:8gaussian_kl_sample}(b)-(k) display the fitting results for the target distribution by QD3PM and the factorization method for $N$=4, 6, 8, 9, and 10. In each subfigure, the orange dashed line represents the probability mass function of the target distribution. It is evident that QD3PM (first row) performs well in all cases, closely approximating the target distribution. Conversely, the factorization method consistently shows insufficient peak heights and higher probabilities in the tails of the x-axis compared to the target distribution. The distributions learned by factorization methods are more uniform than those learned by QD3PM, a consequence of their failure to model the joint distribution. These results highlight the advantages of leveraging quantum computing for modeling joint distributions.

\textbf{Bars-and-Stripes (BAS) Dataset.}
%	On the mixed Gaussian dataset, QD3PM demonstrated an advantage by modeling the joint distribution, outperforming factorization-based methods that fail to capture such correlations. 
To further demonstrate the advantages of modeling the joint distribution, we conducted numerical experiments on the BAS dataset, which is characterized by stronger inter-dimensional correlations than the mixed Gaussian dataset.

\cref{fig:N_KL_compare}(a) showcases the KL divergence trends between QD3PM and its factorization-based variant across varying problem sizes $N$. Error bars represent standard deviations, and dashed lines depict the fitted linear equations. Subfigures (b)-(p) display the BAS target distribution (first row, orange bars), the distributions generated by QD3PM post-training (second row, blue bars), and those generated by the factorization method (third row, pink bars) for different qubit count $N$.

\begin{figure*}[t]
	\centering
	\includegraphics[width=\linewidth]{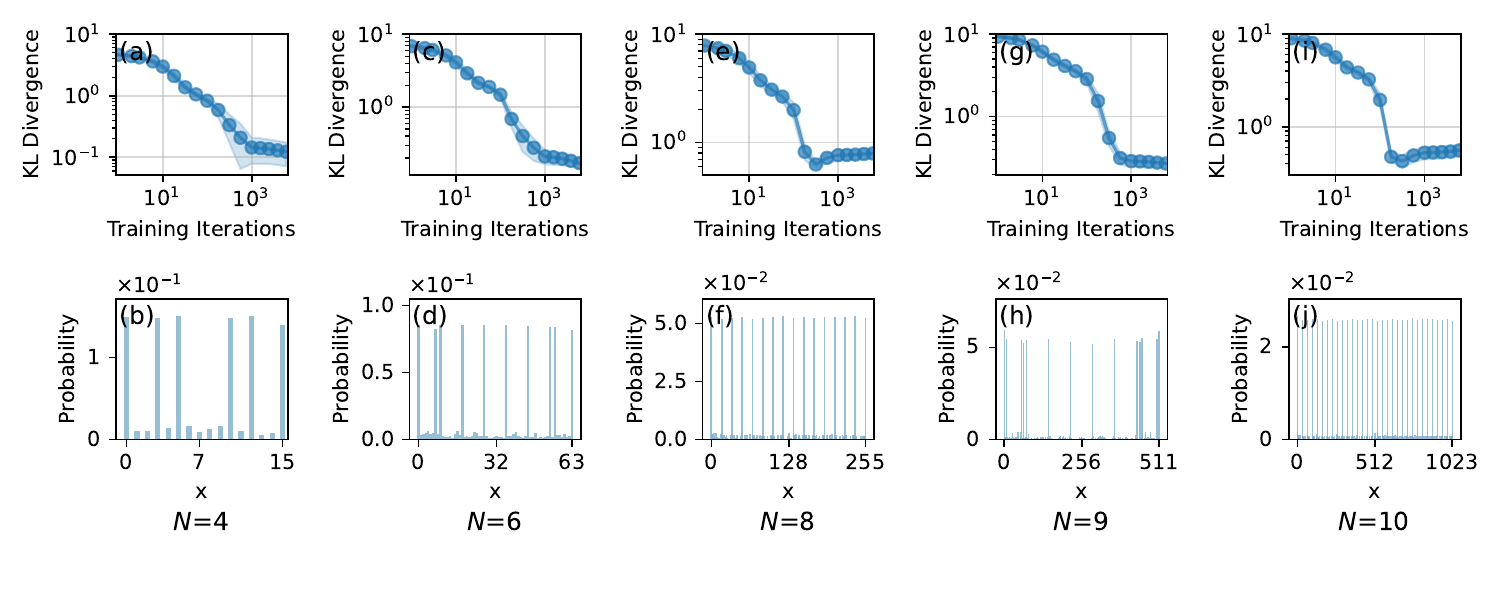}
	\caption{Performance evaluation of the QD3PM denoising circuit for one-step inference on the BAS generation tasks. The first row depicts the convergence trajectories of the KL divergence between the learned and target distributions throughout training. The second row visualizes the distribution obtained via one-step inference from the trained circuit, represented by $p_{\boldsymbol{\theta}}(\boldsymbol{x}_0|\boldsymbol{x}_t)$, at diffusion timestep $t=T$.}
	\label{fig: one-step inference}
\end{figure*}

As shown in \cref{fig:N_KL_compare}(a), QD3PM maintains a low KL divergence value. Benefiting from modeling the joint distribution, the KL error does not exhibit a significant increasing trend with $N$. 
%With the exception of the $N=9$ case. This anomaly at $N=9$ might be attributed to the increased difficulty of the 9-bit BAS dataset, which requires the model to generate 14 specific BAS patterns out of 512 possibilities with uniform probability, a demanding generation task that could hinder a highly precise fit for $N=9$.}
Conversely, the KL fitting divergence for the factorization-based method increases approximately linearly with $N$. This is due to the strong correlations between dimensions in the BAS pattern data. Consequently, the trend of increasing fitting error with $N$ is considerably steeper for the BAS dataset compared to the mixed Gaussian distribution results shown in \cref{fig:8gaussian_kl_sample}(a). This is because BAS data exhibits stronger internal correlations than mixed Gaussian data. This result finding aligns with our theoretical conclusions in \cref{thm:KL_bounds}.

The visualizations in Figures \ref{fig:N_KL_compare}(b)-(p) further highlight QD3PM's advantage in fitting the BAS distribution compared to the factorization method, because it models the joint distribution. Notably, QD3PM maintains high fitting accuracy even at $N=10$. In contrast, the factorization-based method, while sometimes assigning higher probabilities to specific BAS patterns, frequently generates non-BAS patterns. This deficiency arises from its failure to model the joint distribution.

% \textcolor{blue}{\cref{fig:N_KL_compare}(a) illustrates the KL divergence trends between QD3PM and its factorization-based variant across varying problem sizes. Error bars represent standard deviations, and dashed lines represent the fitted linear equations. QD3PM exhibits superior approximation capacity due to its ability to model the joint distribution directly, thereby capturing cross-dimensional dependencies effectively. In contrast, factorization methods, which compute transition probabilities independently across each dimension, suffer from linearly increasing approximation errors as dimensionality grows. These empirical results corroborate our theoretical analysis established in Theorem 3.2. \cref{fig:N_KL_compare}(c) displays the target distribution. \cref{fig:N_KL_compare}(c) displays QD3PM's sampling distribution for the $N=10$ BAS problem, while \cref{fig:N_KL_compare}(d) shows that of the factorization method. QD3PM's distribution is a noticeably better match to the target distribution than the factorization method's.}

\subsubsection{Training Circuit as $p_{\boldsymbol{\theta}}(\boldsymbol{x}_{0}|\boldsymbol{x}_t)$}

\cref{fig: one-step inference} evaluates our QD3PM on the BAS generation task across $N=4, 6, 8, 9$, and 10. The first row displays the training convergence trajectories of the KL divergence. The second row displays the $\boldsymbol{x}_0$ distribution generated by one-step inference from the pure noise distribution $p(\boldsymbol{x}_T)$, using the trained quantum circuit that models $p_{\boldsymbol{\theta}}(\boldsymbol{x}_0|\boldsymbol{x}_T)$. Notably, QD3PM effectively fits the target distribution even for $N=10$. This successful one-step generation demonstrates the circuit's capability to capture the data's underlying structural features and distribution, which is facilitated by QD3PM's preservation of intrinsic data correlations during training.

% \begin{table*}[htbp]
% \centering
% \caption{}
% \begin{tabular}{|c|c|c|c|c|c|}
% \hline
% N             & 4             & 6             & 8             & 9             & 10            \\ \hline
% KL & 0.12$\pm$0.05 & 0.17$\pm$0.03 & 0.79$\pm$0.08 & 0.27$\pm$0.03 & 0.56$\pm$0.08 \\ \hline
% Accuracy      & 0.89$\pm$0.04 & 0.85$\pm$0.02 & 0.84$\pm$0.01 & 0.77$\pm$0.03 & 0.83$\pm$0.01 \\ \hline
% \end{tabular}
% \end{table*}

\section{Discussion} \label{section: conclusion}

We have developed the Quantum Discrete Denoising Diffusion Probabilistic Model (QD3PM) to address the intrinsic limitations of classical discrete diffusion models in handling high-dimensional correlations. By rigorously analyzing the dimension-factorized approach in existing methods, we have proved that their worst-case KL divergence scales as $\mathcal{O}(N\log K)$ with data dimensionality---a fundamental barrier arising from independent per-dimension probability modeling. QD3PM overcomes this limitation through quantum-enhanced joint distribution learning, where the exponentially large Hilbert space inherently preserves inter-dimensional correlations while avoiding the computational bottlenecks with the existing models.

Within the quantum framework, we have established three key methodological advances: 1) A quantum Bayes' theorem formulation for posterior derivation that fundamentally differs from existing dimension-wise computation, enabling accurate modeling of joint probability evolution during diffusion; 2) An efficient quantum denoising architecture combining temporal parameter sharing with classical-data-controlled rotation gates, achieving effective information encoding while maintaining practical parameter efficiency; and 3) A framework for training quantum circuits for single-step sampling, established by leveraging the linearity of quantum operations and quantum measurement properties, thereby avoiding inefficient iterative sampling during the generation. Our simulations of up to 10-qubit systems demonstrate QD3PM's superior distribution fitting accuracy compared to the factorizing-based methods. 
% While the current work focuses on establishing theoretical advantages over classical discrete diffusion models, future works will extend this framework through systematic comparisons with other classical and quantum models under unified benchmarking conditions.

Our work explores the potential of quantum computing to address the dimensional factorization limitations in classical discrete diffusion models. Through the QD3PM framework, we demonstrate that quantum-enhanced joint distribution modeling can effectively preserve inter-dimensional correlations, which are often lost in classical approaches. Limited by the difficulty of classically simulating quantum systems, our preliminary results are based on systems of up to 10 qubits, yet they suggest that quantum generative models hold promise for improving discrete diffusion-based generation. We hope that this framework can serve as a theoretical foundation and a practical stepping stone for further investigations into quantum-enhanced generative learning.

\section{Methods}

\subsection{Proof of  \cref{thm:KL_bounds}}
\label{section: proof of KL}

\begin{proof}
	
	The validity of \cref{eq:kl with quantum} is evident, as our quantum model directly captures the joint probability distribution $Q_{q}(\boldsymbol{x}) := Q_{q}(x^{(1)}, \dots, x^{(N)})$. Ideally, this model can accurately fit the target distribution $P$, ensuring the correctness of the equation.  
	Next, we prove the \cref{eq: kl with classical}:

	\textbf{1.} Given $ Q_c(\boldsymbol{x}) = \prod_{i=1}^{N} Q_c(x^{(i)}) $ and the assumption $ Q_c(x^{(i)}) = P(x^{(i)}) $, we have:
	\begin{equation}  
		Q_c(\boldsymbol{x}) = \prod_{i=1}^{N} P(x^{(i)}).
	\end{equation}
	Substituting into the KL divergence:
	\begin{equation}
		\label{eq: DKL_P_Qc}
		\begin{split}
			D_{\text{KL}}(P \| Q_c) &= \sum_{\boldsymbol{x}} P(\boldsymbol{x}) \log \frac{P(\boldsymbol{x})}{\prod_{i=1}^{N} P(x^{(i)})}.\\
			&= \sum_{\boldsymbol{x}} P(\boldsymbol{x}) \log P(\boldsymbol{x}) - \sum_{\boldsymbol{x}} P(\boldsymbol{x}) \sum_{i=1}^{N} \log P(x^{(i)})\\
			&= -H(P(\boldsymbol{x})) - \sum_{i=1}^{N} \sum_{\boldsymbol{x}} P(\boldsymbol{x}) \log P(x^{(i)}),
		\end{split}
	\end{equation} 
	where $ H(P(\boldsymbol{x})) = \sum_{\boldsymbol{x}} P(\boldsymbol{x}) \log \frac{1}{P(\boldsymbol{x})} $ is the entropy of $ P(\boldsymbol{x}) $. Recognizing that:
	\begin{equation}
		\begin{split}
					\sum_{\boldsymbol{x}} P(\boldsymbol{x}) \log P(x^{(i)}) &= \sum_{x^{(i)}} P(x^{(i)}) \log P(x^{(i)}) \\
					&= -H(P(x^{(i)})).
		\end{split}
	\end{equation}
	Thus, we have:
	\begin{equation}
		D_{\text{KL}}(P \| Q_c) = -H(P(\boldsymbol{x})) + \sum_{i=1}^{N} H(P(x^{(i)})).
	\end{equation}
	By the definition of mutual information:
	\begin{equation}
		I(x^{(1)}; x^{(2)}; \cdots; x^{(N)}) = \sum_{i=1}^{N} H(P(x^{(i)})) - H(P(\boldsymbol{x})).
	\end{equation}
	Thus:
	\begin{equation} \begin{split} 
			D_{\text{KL}}(P \| Q_c) &= I(x^{(1)}; x^{(2)}; \cdots; x^{(N)})\\
			&= \sum_{i=1}^{N} H(P(x^{(i)})) - H(P(\boldsymbol{x})).
	\end{split} \end{equation}

	\textbf{2.} Mutual information $ I(X_1; X_2; \cdots; X_N) $ is maximized when the sum of marginal entropies is maximized and the joint entropy is minimized. For each dimension $ i $, the marginal entropy $ H(P(x^{(i)})) $ will reach its maximum when $ P(x^{(i)}) $ is a uniform distribution. In this case, the marginal entropy has an upper bound:
	\begin{equation}
		H(P(x^{(i)})) = \log K.
	\end{equation}
	Therefore, the sum of marginal entropies is:
	\begin{equation}
		\sum_{i=1}^{N} H(P(x^{(i)})) =  N \log K.
	\end{equation}
	Given that $ P(\boldsymbol{x}) = \frac{1}{K}, \boldsymbol{x} \in \mathcal{D}_{\text{FC}} $, we have:
	\begin{equation}
		H(P(\boldsymbol{x})) = \log K.
	\end{equation}
	Thus, we have:
	\begin{equation} \begin{split} 
			D_{\text{KL}}(P \| Q_c) &= I(x^{(1)}; x^{(2)}; \cdots; x^{(N)}) \\
			&= \sum_{i=1}^{N} H(x^{(i)}) - H(P) \\ 
			&= N \log K - \log K \\
			&= (N - 1) \log K.
	\end{split} \end{equation} 
%	Therefore, the upper bound of $ D_{\text{KL}}(P \| Q_c) $ is:
%	\begin{equation}
%		D_{\text{KL}}(P \| Q_c) = (N - 1) \log K.
%	\end{equation}
	This completes the proof.
\end{proof}

%\textcolor{blue}{  
%	\textbf{Relation to Prior Work.}  
%	We observe that the term  
%	$
%	\sum_{\boldsymbol{x}} P(\boldsymbol{x}) \log\frac{P(\boldsymbol{x})}{\prod_{i=1}^N P(x^{(i)})}
%	$
%	in our ~\cref{eq: DKL_P_Qc} coincides with the irreducible term in Proposition 1 of Liu et al.~\cite{liu2025discrete}. This overlap stems from both analyses beginning with the same factorization assumption, which inevitably produces identical preliminary steps. However, our objectives diverge:
%	\begin{itemize}
%		\item Our work derives the exact KL divergence between fully correlated and factorized joint distributions, explicitly quantifying the worst-case fitting error and its scaling with data dimensions. 
%		\item \cite{liu2025discrete} analyzes how factorization introduces a gap in the ELBO for diffusion models, highlighting its impact on optimization. 
%	\end{itemize}
%	Therefore, both works expose the intrinsic limitations of factorization in discrete diffusion models, but through complementary perspectives.
%}

\subsection{Proof of  \cref{thm:non-unitary posterior}}
\label{section: proof of posterior}

% \textbf{\cref{thm:non-unitary posterior}}
% \textit{
	% 	If a diffusion process is a depolarizing channel parameterized by $\alpha_t$ in (\ref{eq:one step depolarizing}), and all density matrices during diffusion are diagonal, the posterior state used for training the denoising model, representing the ground-truth posterior distribution, can be derived from the quantum conditional state formalism and quantum Bayes' theorem as follows:
	% 	\begin{equation}
		% 		\begin{split} 
			% 			\tilde{\rho}_{t-1 | X_t = \boldsymbol{x}_t, X_0 = \boldsymbol{x}_0} 
			% 			= \frac{\alpha_t \langle \boldsymbol{x}_t | \rho_{t-1} | \boldsymbol{x}_t \rangle |\boldsymbol{x}_t\rangle \langle \boldsymbol{x}_t| + \frac{1}{d^2}(1-\alpha_t) \rho_{t-1}}{\alpha_t \langle \boldsymbol{x}_t | \rho_{t-1} | \boldsymbol{x}_t \rangle + \frac{1}{d^2} (1-\alpha_t)},
			% 		\end{split} 
		% 	\end{equation}
	% 	where $ \rho_{t-1} = (1 - \bar{\alpha}_{t-1}) \boldsymbol{I}/d + \bar{\alpha}_{t-1} |\boldsymbol{x}_0\rangle \langle \boldsymbol{x}_0| $.}

\begin{proof}
	1. Given that $ \rho_t = (1 - \alpha_t) \boldsymbol{I}/d + \alpha_t \rho_{t-1} $, we can express the Choi-Jamio\l{}kowski isomorphic operator of the depolarizing channel as:
	\begin{equation}
		\varrho_{\mathcal{H}_t|\mathcal{H}_{t-1}} = (\mathcal{E}_{\mathcal{H}_{t}|\mathcal{H}_{t-1}^{'}} \otimes \boldsymbol{I}_{\mathcal{H}_{t-1}} )(|\Phi^{+}\rangle\langle\Phi^{+}|_{\mathcal{H}^{'}_{t-1}\mathcal{H}_{t-1}}).
	\end{equation}
	Since $ |\Phi^{+}\rangle\langle\Phi^{+}|_{\mathcal{H}^{'}_{t-1}\mathcal{H}_{t-1}} = \frac{1}{d} \sum_{\boldsymbol{i},\boldsymbol{j}=0}^{d-1} |\boldsymbol{i}\rangle \langle \boldsymbol{j}|_{\mathcal{H}^{'}_{t-1}} \otimes |\boldsymbol{i}\rangle \langle \boldsymbol{j}|_{\mathcal{H}_{t-1}} $, and by the definition of a maximally entangled state, the reduced density matrix on the subsystem $ \mathcal{H}_{t-1} $ is:
	\begin{equation}
		\text{Tr}_{\mathcal{H}^{'}_{t-1}}\left(|\Phi^{+}\rangle\langle\Phi^{+}|_{\mathcal{H}^{'}_{t-1}\mathcal{H}_{t-1}}\right) = \frac{\boldsymbol{I}_{\mathcal{H}_{t-1}}}{d}.
	\end{equation}
	Therefore, the operator $ \varrho_{\mathcal{H}_t|\mathcal{H}_{t-1}} $ simplifies to:
	\begin{equation}
		\varrho_{\mathcal{H}_t|\mathcal{H}_{t-1}} = \alpha_t |\Phi^{+}\rangle\langle \Phi^{+}|_{\mathcal{H}_t \mathcal{H}_{t-1}} + (1 - \alpha_t) \frac{\boldsymbol{I}_{\mathcal{H}_{t}}}{d} \otimes \frac{\boldsymbol{I}_{\mathcal{H}_{t-1}}}{d}.
	\end{equation}
	
	2. Since our work involves measuring $ \rho_t $ in the computational basis at each timestep, the process of obtaining the random variable $ X_t $ from $ \rho_t $ can be described by a positive causal conditional state operator:
	\begin{equation}
		\varrho_{X_t|\mathcal{H}_{t}} = \sum_{\boldsymbol{x}=0}^{d-1} |\boldsymbol{x}\rangle\langle \boldsymbol{x}|_{X_t} \otimes |\boldsymbol{x}\rangle\langle \boldsymbol{x}|_{\mathcal{H}_t}.
	\end{equation}
	
	3. Therefore, the causal conditional state from $ \rho_{t-1} $ to $ X_t $ is:
	\begin{equation} 
		\begin{split}
			\varrho_{X_t|\mathcal{H}_{t-1}} 
			&= \text{Tr}_{\mathcal{H}_t}\left( \varrho_{X_t|\mathcal{H}_t} \cdot \varrho_{\mathcal{H}_t|\mathcal{H}_{t-1}} \right) \\
			&= \text{Tr}_{\mathcal{H}_t}\Bigg[ \sum_{\boldsymbol{x}=0}^{d-1} |\boldsymbol{x}\rangle \langle \boldsymbol{x}|_{X_t} 
			\otimes \Bigg( \frac{\alpha_t}{d} \sum_{\boldsymbol{i},\boldsymbol{j}=0}^{d-1} |\boldsymbol{i}\rangle \langle \boldsymbol{j}|_{\mathcal{H}_t}  \\
			& \otimes |\boldsymbol{i}\rangle \langle \boldsymbol{j}|_{\mathcal{H}_{t-1}}+ (1 - \alpha_t) \frac{\boldsymbol{I}_{\mathcal{H}_t}}{d} \otimes \frac{\boldsymbol{I}_{\mathcal{H}_{t-1}}}{d} \Bigg) \Bigg]
		\end{split}
	\end{equation}
%\begin{equation}
%	\begin{split}
%		&\varrho_{X_t|\mathcal{H}_{t-1}} \\
%		= &\text{Tr}_{\mathcal{H}_t}\left( \varrho_{X_t|\mathcal{H}_t} \cdot \varrho_{\mathcal{H}_t|\mathcal{H}_{t-1}} \right) \\
%		= &\text{Tr}_{\mathcal{H}_t}\left[ \sum_{\boldsymbol{x}=0}^{d-1} |\boldsymbol{x}\rangle \langle \boldsymbol{x}|_{X_t} \otimes \left( \frac{\alpha_t}{d} \sum_{\boldsymbol{i},\boldsymbol{j}=0}^{d-1} |\boldsymbol{i}\rangle \langle \boldsymbol{j}|_{\mathcal{H}_t} \otimes |\boldsymbol{i}\rangle \langle \boldsymbol{j}|_{\mathcal{H}_{t-1}} \right. \right. \right. \\
%		& \left. \left. + (1 - \alpha_t) \frac{\boldsymbol{I}_{\mathcal{H}_t}}{d} \otimes \frac{\boldsymbol{I}_{\mathcal{H}_{t-1}}}{d} \right)  \right].
%	\end{split}
%\end{equation}
	where
	\begin{equation} \begin{split} 
			&\left( \sum_{\boldsymbol{x}=0}^{d-1} |\boldsymbol{x}\rangle\langle \boldsymbol{x}|_{X_t} \otimes |\boldsymbol{x}\rangle\langle \boldsymbol{x}|_{\mathcal{H}_t} \right) \left( \frac{\alpha_t}{d} \sum_{\boldsymbol{i},\boldsymbol{j}=0}^{d-1} |\boldsymbol{i}\rangle \langle \boldsymbol{j}|_{\mathcal{H}_t} \otimes |\boldsymbol{i}\rangle \langle \boldsymbol{j}|_{\mathcal{H}_{t-1}} \right) \\
			&= \sum_{\boldsymbol{x},\boldsymbol{i},\boldsymbol{j}=0}^{d-1} |\boldsymbol{x}\rangle\langle \boldsymbol{x}|_{X_t} \otimes \frac{\alpha_t}{d} |\boldsymbol{x}\rangle\langle \boldsymbol{x}| \boldsymbol{i}\rangle \langle \boldsymbol{j}|_{\mathcal{H}_t} \otimes |\boldsymbol{i}\rangle \langle \boldsymbol{j}|_{\mathcal{H}_{t-1}} \\
			&= \sum_{\boldsymbol{x},\boldsymbol{i},\boldsymbol{j}=0}^{d-1} |\boldsymbol{x}\rangle\langle \boldsymbol{x}|_{X_t} \otimes \delta_{\boldsymbol{x},\boldsymbol{i}} \frac{\alpha_t}{d} |\boldsymbol{x}\rangle \langle \boldsymbol{j}|_{\mathcal{H}_t} \otimes |\boldsymbol{x}\rangle \langle \boldsymbol{j}|_{\mathcal{H}_{t-1}} \\
			&= \sum_{\boldsymbol{x},\boldsymbol{j}=0}^{d-1} |\boldsymbol{x}\rangle\langle \boldsymbol{x}|_{X_t} \otimes \frac{\alpha_t}{d} |\boldsymbol{x}\rangle \langle \boldsymbol{j}|_{\mathcal{H}_t} \otimes |\boldsymbol{x}\rangle \langle \boldsymbol{j}|_{\mathcal{H}_{t-1}}.
	\end{split} \end{equation} 
	Another part is
	\begin{equation} \begin{split} 
			&\left( \sum_{\boldsymbol{x}=0}^{d-1} |\boldsymbol{x}\rangle\langle \boldsymbol{x}|_{X_t} \otimes |\boldsymbol{x}\rangle\langle \boldsymbol{x}|_{\mathcal{H}_t} \right) \cdot \left( (1 - \alpha_t) \frac{\boldsymbol{I}_{\mathcal{H}_{t}}}{d} \otimes \frac{\boldsymbol{I}_{\mathcal{H}_{t-1}}}{d} \right) \\
			&= \frac{(1 - \alpha_t)}{d^2} \sum_{\boldsymbol{x}=0}^{d-1} |\boldsymbol{x}\rangle\langle \boldsymbol{x}|_{X_t} \otimes |\boldsymbol{x}\rangle\langle \boldsymbol{x}|_{\mathcal{H}_t} \otimes \boldsymbol{I}_{\mathcal{H}_{t-1}}.
	\end{split} \end{equation} 
	Therefore, we obtain:
	\begin{equation} 
		\begin{split}
			\varrho_{X_t|\mathcal{H}_{t-1}} 
			&= \text{Tr}_{\mathcal{H}_t} \Biggl[ 
			\sum_{\boldsymbol{x},\boldsymbol{j}=0}^{d-1} |\boldsymbol{x}\rangle\langle \boldsymbol{x}|_{X_t} \otimes \frac{\alpha_t}{d} |\boldsymbol{x}\rangle \langle \boldsymbol{j}|_{\mathcal{H}_t} \otimes |\boldsymbol{x}\rangle \langle \boldsymbol{j}|_{\mathcal{H}_{t-1}} \\
			&\quad + \frac{(1 - \alpha_t)}{d^2} \sum_{\boldsymbol{x}=0}^{d-1} |\boldsymbol{x}\rangle\langle \boldsymbol{x}|_{X_t} \otimes |\boldsymbol{x}\rangle\langle \boldsymbol{x}|_{\mathcal{H}_t} \otimes \boldsymbol{I}_{\mathcal{H}_{t-1}} \Biggr] \\
			&= \frac{\alpha_t}{d} \sum_{\boldsymbol{x}=0}^{d-1} |\boldsymbol{x}\rangle\langle \boldsymbol{x}|_{X_t} \otimes |\boldsymbol{x}\rangle\langle \boldsymbol{x}|_{\mathcal{H}_{t-1}} \\
			&\quad + \frac{(1 - \alpha_t)}{d^2} \sum_{\boldsymbol{x}=0}^{d-1} |\boldsymbol{x}\rangle\langle \boldsymbol{x}|_{X_t} \otimes \boldsymbol{I}_{\mathcal{H}_{t-1}} \\
			&= \sum_{\boldsymbol{x}=0}^{d-1} |\boldsymbol{x}\rangle\langle \boldsymbol{x}|_{X_t} \otimes \left( \frac{\alpha_t}{d} |\boldsymbol{x}\rangle\langle \boldsymbol{x}|_{\mathcal{H}_{t-1}} + \frac{(1 - \alpha_t)}{d^2} \boldsymbol{I}_{\mathcal{H}_{t-1}} \right).
		\end{split} 
	\end{equation}
%\begin{equation} \begin{split} 
%		\varrho_{X_t|\mathcal{H}_{t-1}} &= \text{Tr}_{\mathcal{H}_t} \left[ \sum_{\boldsymbol{x},\boldsymbol{j}=0}^{d-1} |\boldsymbol{x}\rangle\langle \boldsymbol{x}|_{X_t} \otimes \frac{\alpha_t}{d} |\boldsymbol{x}\rangle \langle \boldsymbol{j}|_{\mathcal{H}_t} \otimes |\boldsymbol{x}\rangle \langle \boldsymbol{j}|_{\mathcal{H}_{t-1}} \\
%		& \quad  + \frac{(1 - \alpha_t)}{d^2} \sum_{\boldsymbol{x}=0}^{d-1} |\boldsymbol{x}\rangle\langle \boldsymbol{x}|_{X_t} \otimes |\boldsymbol{x}\rangle\langle \boldsymbol{x}|_{\mathcal{H}_t} \otimes \boldsymbol{I}_{\mathcal{H}_{t-1}} \right] \\
%		&= \frac{\alpha_t}{d} \sum_{\boldsymbol{x}=0}^{d-1} |\boldsymbol{x}\rangle\langle \boldsymbol{x}|_{X_t} \otimes |\boldsymbol{x}\rangle\langle \boldsymbol{x}|_{\mathcal{H}_{t-1}} \\
%		& \quad + \frac{(1 - \alpha_t)}{d^2} \sum_{\boldsymbol{x}=0}^{d-1} |\boldsymbol{x}\rangle\langle \boldsymbol{x}|_{X_t} \otimes \boldsymbol{I}_{\mathcal{H}_{t-1}} \\
%		&= \sum_{\boldsymbol{x}=0}^{d-1} |\boldsymbol{x}\rangle\langle \boldsymbol{x}|_{X_t} \otimes \left( \frac{\alpha_t}{d} |\boldsymbol{x}\rangle\langle \boldsymbol{x}|_{\mathcal{H}_{t-1}} + \frac{(1 - \alpha_t)}{d^2} \boldsymbol{I}_{\mathcal{H}_{t-1}} \right).
%\end{split} \end{equation} 
	Given that under the condition $ X_t = \boldsymbol{x}_t $:
	\begin{equation}
		\varrho_{X_t=\boldsymbol{x}_t|\mathcal{H}_{t-1}} = \frac{\alpha_t}{d} |\boldsymbol{x}_t\rangle\langle \boldsymbol{x}_t|_{\mathcal{H}_{t-1}} + \frac{(1 - \alpha_t)}{d^2} \boldsymbol{I}_{\mathcal{H}_{t-1}}.
	\end{equation}
	The probability distribution of the measurement result $ X_t $ is determined by $ \varrho_{X_t|\mathcal{H}_{t-1}} $ and $ \rho_{t-1} $:
	\begin{equation} \begin{split} 
			\rho_{X_t} &= \text{Tr}_{\mathcal{H}_{t-1}} \left[ \varrho_{X_t|\mathcal{H}_{t-1}} \cdot (\boldsymbol{I}_{X_t} \otimes \rho_{t-1}) \right] \\
			&= \text{Tr}_{\mathcal{H}_{t-1}} \left[ \sum_{\boldsymbol{x}=0}^{d-1} |\boldsymbol{x}\rangle\langle \boldsymbol{x}|_{X_t} \otimes \left( \frac{\alpha_t}{d} |\boldsymbol{x}\rangle\langle \boldsymbol{x}| \rho_{t-1} + \frac{(1 - \alpha_t)}{d^2} \rho_{t-1} \right) \right] \\
			&= \sum_{\boldsymbol{x}=0}^{d-1} \left( \alpha_t \langle \boldsymbol{x} | \rho_{t-1} | \boldsymbol{x} \rangle + \frac{1 - \alpha_t}{d^2} \right) |\boldsymbol{x}\rangle\langle \boldsymbol{x}| \\
			&= \sum_{\boldsymbol{x}=0}^{d-1} p(X_t = \boldsymbol{x}) |\boldsymbol{x}\rangle\langle \boldsymbol{x}|.
	\end{split} \end{equation} 
	For each measurement result $ X_t = \boldsymbol{x}_t $, the probability is:
	\begin{equation} \begin{split} 
			\rho_{X_t=\boldsymbol{x}_t} &= p(X_t = \boldsymbol{x}_t) \\
			&= \alpha_t \langle \boldsymbol{x}_t | \rho_{t-1} | \boldsymbol{x}_t \rangle + \frac{1 - \alpha_t}{d^2}.
	\end{split} \end{equation} 
	
	4. Therefore, we obtain the final result:
	\begin{equation} \begin{split}  
			\tilde{\rho}_{t-1|X_t=\boldsymbol{x}_t} &= \varrho_{X_t=\boldsymbol{x}_t|\mathcal{H}_{t-1}} \star \left( \rho_{t-1} \rho_{X_t=\boldsymbol{x}_t}^{-1} \right) \\
			&= \frac{ \sqrt{\rho_{t-1}} \left( \alpha_t |\boldsymbol{x}_t\rangle\langle \boldsymbol{x}_t| + \frac{(1 - \alpha_t)}{d^2} \boldsymbol{I} \right) \sqrt{\rho_{t-1}} }{ \alpha_t \langle \boldsymbol{x}_t | \rho_{t-1} | \boldsymbol{x}_t \rangle + (1 - \alpha_t) \frac{1}{d^2} } \\   
%			&= \frac{ \sum_{\boldsymbol{y}} \sqrt{ \langle \boldsymbol{y} | \rho_{t-1} | \boldsymbol{y} \rangle } |\boldsymbol{y}\rangle\langle \boldsymbol{y}| \left( \alpha_t |\boldsymbol{x}_t\rangle\langle \boldsymbol{x}_t| + \frac{(1 - \alpha_t)}{d^2} \boldsymbol{I} \right) \sum_{\boldsymbol{y}} \sqrt{ \langle \boldsymbol{y} | \rho_{t-1} | \boldsymbol{y} \rangle } |\boldsymbol{y}\rangle\langle \boldsymbol{y}| }{ \alpha_t \langle \boldsymbol{x}_t | \rho_{t-1} | \boldsymbol{x}_t \rangle + (1 - \alpha_t) \frac{1}{d^2} } \\
%			&= \frac{ \alpha_t \langle \boldsymbol{x}_t | \rho_{t-1} | \boldsymbol{x}_t \rangle |\boldsymbol{x}_t\rangle\langle \boldsymbol{x}_t| + \frac{(1 - \alpha_t)}{d^2} \sum_{\boldsymbol{y}} \langle \boldsymbol{y} | \rho_{t-1} | \boldsymbol{y} \rangle |\boldsymbol{y}\rangle\langle \boldsymbol{y}| }{ \alpha_t \langle \boldsymbol{x}_t | \rho_{t-1} | \boldsymbol{x}_t \rangle + (1 - \alpha_t) \frac{1}{d^2} } \\ 
			&= \frac{ \alpha_t \langle \boldsymbol{x}_t | \rho_{t-1} | \boldsymbol{x}_t \rangle |\boldsymbol{x}_t\rangle\langle \boldsymbol{x}_t| + \frac{(1 - \alpha_t)}{d^2} \rho_{t-1} }{ \alpha_t \langle \boldsymbol{x}_t | \rho_{t-1} | \boldsymbol{x}_t \rangle + (1 - \alpha_t) \frac{1}{d^2} }.
	\end{split} \end{equation}  
This completes the proof.
	
\end{proof}

\subsection{Derivation of the Circuit Architecture to Generate the Posterior State.}
\label{section: pose circuit derivation}
%\begin{figure}[h]
%	\centering
%	\includegraphics[width=0.7\linewidth]{image/pos_circuit-eps-converted-to.pdf}
%	\caption{Circuit architecture to generate the posterior state.}
%	\label{fig:pos_circuit}
%\end{figure}

Here, we prove that the circuit shown in \cref{fig:pos_circuit1} yields our target ground-truth posterior state. The composite state after applying the $ R_y $ gate is given by:

Set:
\begin{equation}
	\phi = 2 \arccos \left( \sqrt{ \frac{ \alpha_t \langle \boldsymbol{x}_t | \rho_{t-1} | \boldsymbol{x}_t \rangle }{ \alpha_t \langle \boldsymbol{x}_t | \rho_{t-1} | \boldsymbol{x}_t \rangle + (1 - \alpha_t) \frac{1}{d^2} } } \right).
\end{equation}
The input quantum state is:
\begin{equation}
	|0\rangle\langle 0| \otimes \rho_{t-1} \otimes |\boldsymbol{x}_t\rangle\langle \boldsymbol{x}_t|.
\end{equation}
After applying the $ R_y $ gate, the quantum state becomes:
\begin{equation}
	\begin{split}
	&\cos^2\left(\frac{\phi}{2}\right) |0\rangle\langle 0| \otimes \rho_{t-1} \otimes |\boldsymbol{x}_t\rangle\langle \boldsymbol{x}_t| \\
	&+ \sin^2\left(\frac{\phi}{2}\right) |1\rangle\langle 1| \otimes \rho_{t-1} \otimes |\boldsymbol{x}_t\rangle\langle \boldsymbol{x}_t| + \cdots,
	\end{split}
\end{equation}
where the \textquotedblleft$\cdots$\textquotedblright~terms are omitted because the subsequent partial trace operation will exclude them from the result.
After applying the controlled swap gate, the quantum state becomes:
\begin{equation}
	\begin{split}
	&\cos^2\left(\frac{\phi}{2}\right) |0\rangle\langle 0| \otimes \rho_{t-1} \otimes |\boldsymbol{x}_t\rangle\langle \boldsymbol{x}_t|\\
	 &+ \sin^2\left(\frac{\phi}{2}\right) |1\rangle\langle 1| \otimes |\boldsymbol{x}_t\rangle\langle \boldsymbol{x}_t| \otimes \rho_{t-1} + \cdots.  
\end{split}
\end{equation}  
Taking the partial trace over registers A and B, we obtain:
\begin{equation}
	\begin{split}
		&\text{Tr}_{A,B}\left( \cos^2\left(\frac{\phi}{2}\right) |0\rangle\langle 0| \otimes \rho_{t-1} \otimes |\boldsymbol{x}_t\rangle\langle \boldsymbol{x}_t| \right. \\
		&\quad\quad\quad \left. {} + \sin^2\left(\frac{\phi}{2}\right) |1\rangle\langle 1| \otimes |\boldsymbol{x}_t\rangle\langle \boldsymbol{x}_t| \otimes \rho_{t-1} + \cdots \right) \\
		&= \cos^2\left(\frac{\phi}{2}\right) |\boldsymbol{x}_t\rangle\langle \boldsymbol{x}_t| + \sin^2\left(\frac{\phi}{2}\right) \rho_{t-1}.
	\end{split}
\end{equation}

Thus, we have:
%\begin{equation} \begin{split} 
%		&\text{Tr}_{A,B}\left( \cos^2\left(\frac{\phi}{2}\right) |0\rangle\langle 0| \otimes \rho_{t-1} \otimes |\boldsymbol{x}_t\rangle\langle \boldsymbol{x}_t| \\
%		&+ \sin^2\left(\frac{\phi}{2}\right) |1\rangle\langle 1| \otimes |\boldsymbol{x}_t\rangle\langle \boldsymbol{x}_t| \otimes \rho_{t-1} + \cdots \right) \\
%		&= \cos^2\left(\frac{\phi}{2}\right) |\boldsymbol{x}_t\rangle\langle \boldsymbol{x}_t| + \left(1 - \cos^2\left(\frac{\phi}{2}\right)\right) \rho_{t-1} \\
%		&= \frac{ \alpha_t \langle \boldsymbol{x}_t | \rho_{t-1} | \boldsymbol{x}_t \rangle |\boldsymbol{x}_t\rangle\langle \boldsymbol{x}_t| }{ \alpha_t \langle \boldsymbol{x}_t | \rho_{t-1} | \boldsymbol{x}_t \rangle + (1 - \alpha_t) \frac{1}{d^2} } + \frac{ \frac{(1 - \alpha_t)}{d^2} \rho_{t-1} }{ \alpha_t \langle \boldsymbol{x}_t | \rho_{t-1} | \boldsymbol{x}_t \rangle + (1 - \alpha_t) \frac{1}{d^2} } \\
%		&= \frac{ \alpha_t \langle \boldsymbol{x}_t | \rho_{t-1} | \boldsymbol{x}_t \rangle |\boldsymbol{x}_t\rangle\langle \boldsymbol{x}_t| + \frac{(1 - \alpha_t)}{d^2} \rho_{t-1} }{ \alpha_t \langle \boldsymbol{x}_t | \rho_{t-1} | \boldsymbol{x}_t \rangle + (1 - \alpha_t) \frac{1}{d^2} }.
%\end{split}
% \end{equation}  
\begin{equation}
	\begin{split}
		&\text{Tr}_{A,B}\left( \cos^2\left(\frac{\phi}{2}\right) |0\rangle\langle 0| \otimes \rho_{t-1} \otimes |\boldsymbol{x}_t\rangle\langle \boldsymbol{x}_t| \right. \\ % Closing the initial \left( from Tr_A,B
		&\quad \quad \quad \left. + \sin^2\left(\frac{\phi}{2}\right) |1\rangle\langle 1| \otimes |\boldsymbol{x}_t\rangle\langle \boldsymbol{x}_t| \otimes \rho_{t-1} + \cdots \right) \\ % Starting new \left. and closing with the final \right)
		&= \cos^2\left(\frac{\phi}{2}\right) |\boldsymbol{x}_t\rangle\langle \boldsymbol{x}_t| + \left(1 - \cos^2\left(\frac{\phi}{2}\right)\right) \rho_{t-1} \\
		&= \frac{ \alpha_t \langle \boldsymbol{x}_t | \rho_{t-1} | \boldsymbol{x}_t \rangle |\boldsymbol{x}_t\rangle\langle \boldsymbol{x}_t| }{ \alpha_t \langle \boldsymbol{x}_t | \rho_{t-1} | \boldsymbol{x}_t \rangle + (1 - \alpha_t) \frac{1}{d^2} } \\
		& \quad	+  \frac{ \frac{(1 - \alpha_t)}{d^2} \rho_{t-1} }{ \alpha_t \langle \boldsymbol{x}_t | \rho_{t-1} | \boldsymbol{x}_t \rangle + (1 - \alpha_t) \frac{1}{d^2} } \\
		&= \frac{ \alpha_t \langle \boldsymbol{x}_t | \rho_{t-1} | \boldsymbol{x}_t \rangle |\boldsymbol{x}_t\rangle\langle \boldsymbol{x}_t| + \frac{(1 - \alpha_t)}{d^2} \rho_{t-1} }{ \alpha_t \langle \boldsymbol{x}_t | \rho_{t-1} | \boldsymbol{x}_t \rangle + (1 - \alpha_t) \frac{1}{d^2} }.
	\end{split}
\end{equation}

\subsection{Training and Generation Algorithms}

\label{section: training and generation algorithms}
The process of QD3PM training is outlined in Algorithm~\ref{alg:training}. After training, the optimal parameters are used in Algorithm~\ref{alg:generation} to generate clean data $ \boldsymbol{x}_0 $. 
\begin{algorithm}[H]
	\caption{QD3PM training.}  
	\label{alg:training}
	\begin{algorithmic}[1] 
		\STATE \parbox[t]{0.9\linewidth}{
			\textbf{Input:} Number of qubits $N$, total timestep of the diffusion process $T$, completely mixed state $\rho_T=\boldsymbol{I}/d$, batch size $\text{BS}$, maximum iteration number, learning rate $\eta$, trainable parameters $\boldsymbol{\omega} $ and $ \boldsymbol{v}$, noise schedule parameters $\{\bar{\alpha}_t\}_{t=1}^{T}$.
		}
		\REPEAT 
		\STATE Sample a mini-batch data from dataset  $ \{\boldsymbol{x}_0\} \sim \mathcal{D}_{\text{train}} $  
		\STATE $ \mathcal{L} \leftarrow 0 $
		
		\FOR{Each $ \boldsymbol{x}_0 $}
		\STATE $ t\sim \mathcal{U}(2, T) $
		\STATE $|\boldsymbol{x}_0\rangle \leftarrow \boldsymbol{x}_0$
		\STATE $ \rho_{t-1} \leftarrow (1-\bar{\alpha}_{t-1})\boldsymbol{I}/d + \bar{\alpha}_{t-1}|\boldsymbol{x}_{0}\rangle\langle \boldsymbol{x}_{0}| $
		\STATE $ \rho_t \leftarrow (1-\bar{\alpha}_t)\boldsymbol{I}/d + \bar{\alpha}_t|\boldsymbol{x}_{0}\rangle\langle \boldsymbol{x}_{0}| $
		\STATE  Sample a noisy data $\boldsymbol{x}_t\sim q(\boldsymbol{x}_{t}|\boldsymbol{x}_{0})=\text{Tr}(|\boldsymbol{x}_{t}\rangle\langle\boldsymbol{x}_{0}|\rho_t) $
		\STATE $ |\psi(t, \boldsymbol{x}_t)\rangle \leftarrow R_y\left(\frac{2\pi \cdot t}{T+1}\right)|0\rangle \otimes_{i=1}^{N} R_y\left(w_i \cdot x_t^{(i)}\right)|0\rangle $
		\STATE $ |\psi_{\text{out}}\rangle \leftarrow U(\boldsymbol{v})|\psi(t, \boldsymbol{x}_t)\rangle $
		\STATE $ \rho_{\text{out}} \leftarrow \text{Tr}_1(|\psi_{\text{out}}\rangle\langle\psi_{\text{out}}|) $
		\STATE Model distribution $ p_{\boldsymbol{\theta}}(\boldsymbol{x}_{t-1}|\boldsymbol{x}_{t})=\text{Tr}(|\boldsymbol{x}_{t-1}\rangle\langle\boldsymbol{x}_{t-1}|\rho_{\text{out}}) $ \quad $ \# \boldsymbol{\theta}=\{\boldsymbol{w},\boldsymbol{v}\} $
		\STATE Compute posterior state $ \tilde{\rho}_{t-1} 
		\leftarrow \frac{\alpha_t \langle \boldsymbol{x}_t | \rho_{t-1} | \boldsymbol{x}_t \rangle |\boldsymbol{x}_t\rangle \langle \boldsymbol{x}_t| + \frac{1}{d^2}(1-\alpha_t) \rho_{t-1}}{\alpha_t \langle \boldsymbol{x}_t | \rho_{t-1} | \boldsymbol{x}_t \rangle + \frac{1}{d^2} (1-\alpha_t)} $
		\STATE Ground-truth distribution $ q(\boldsymbol{x}_{t-1} | \boldsymbol{x}_t, \boldsymbol{x}_0)=\text{Tr}(|\boldsymbol{x}_{t-1}\rangle\langle\boldsymbol{x}_{t-1}|\tilde{\rho}_{t-1} ) $
		\STATE  $ \mathcal{L}_{t-1}  \leftarrow \mathcal{L}_{\mathrm{MMD}}(p_{\boldsymbol{\theta}}(\boldsymbol{x}_{t-1}|\boldsymbol{x}_{t}),q(\boldsymbol{x}_{t-1} | \boldsymbol{x}_t, \boldsymbol{x}_0))$  
		\STATE Use the procedure outlined in lines 8 through 14 to compute $ \mathcal{L}_{0} \leftarrow \mathcal{L}_{\mathrm{MMD}}(p_{\boldsymbol{\theta}}(\boldsymbol{x}_{0}|\boldsymbol{x}_{1}),\text{Tr}(|\boldsymbol{x}'_0\rangle\langle\boldsymbol{x}'_0||\boldsymbol{x}_0\rangle\langle\boldsymbol{x}_0|))$
		\STATE $ \mathcal{L} \leftarrow \mathcal{L} + \frac{1}{BS}(\mathcal{L}_0+\mathcal{L}_{t-1})$
		\ENDFOR
		\STATE $\boldsymbol{w}\leftarrow \boldsymbol{w} - \eta \nabla_{\boldsymbol{w}}\mathcal{L};\boldsymbol{v} \leftarrow \boldsymbol{v} - \eta \nabla_{\boldsymbol{v}}\mathcal{L}$
		\UNTIL{$\mathcal{L}$ converges or the number of iterations reaches the maximum}
		
		\STATE \textbf{Output:} Optimal parameters $\boldsymbol{\omega}^{*}, \boldsymbol{v}^{*}$ 
	\end{algorithmic}  
\end{algorithm}

\begin{table*}[t]
	\centering
	\caption{The hyperparameters of training circuit as $p_{\boldsymbol{\theta}}(\boldsymbol{x}_{t-1}|\boldsymbol{x}_t)$.}  
	\label{tab: hyperparameters1}  
	\begin{tabular}{|c|cccccccccc|}
		\hline
		Parameters                                                                & \multicolumn{5}{c|}{BAS dataset}                                                                                                     & \multicolumn{5}{c|}{Mixed Gaussian dataset}                                                                   \\ \hline
		$ N $                                                                     & \multicolumn{1}{c|}{4}    & \multicolumn{1}{c|}{6}  & \multicolumn{1}{c|}{8}   & \multicolumn{1}{c|}{9}   & \multicolumn{1}{c|}{10}  & \multicolumn{1}{c|}{4}  & \multicolumn{1}{c|}{6}  & \multicolumn{1}{c|}{8}   & \multicolumn{1}{c|}{9}   & 10  \\ \hline
		$L$                                                                       & \multicolumn{1}{c|}{12}   & \multicolumn{1}{c|}{14} & \multicolumn{1}{c|}{15}  & \multicolumn{1}{c|}{16}  & \multicolumn{1}{c|}{20}  & \multicolumn{1}{c|}{12} & \multicolumn{1}{c|}{12} & \multicolumn{1}{c|}{8}   & \multicolumn{1}{c|}{10}  & 12  \\ \hline
		Batch size                                                                & \multicolumn{1}{c|}{16}   & \multicolumn{1}{c|}{32} & \multicolumn{1}{c|}{128} & \multicolumn{1}{c|}{256} & \multicolumn{1}{c|}{256} & \multicolumn{1}{c|}{32} & \multicolumn{1}{c|}{64} & \multicolumn{1}{c|}{128} & \multicolumn{1}{c|}{128} & 256 \\ \hline
		Number of training data                                                   & \multicolumn{5}{c|}{\textbackslash{}}                                                                                                & \multicolumn{2}{c|}{5000}                         & \multicolumn{3}{c|}{50000}                                \\ \hline
		\multicolumn{1}{|l|}{Number of samples to estimate $p(\boldsymbol{x}_0)$} & \multicolumn{2}{c|}{10000}                          & \multicolumn{3}{c|}{100000}                                                    & \multicolumn{2}{c|}{10000}                        & \multicolumn{3}{c|}{100000}                               \\ \hline
		Learning rate decay steps                                                 & \multicolumn{5}{c|}{3000}                                                                                                            & \multicolumn{5}{c|}{5000}                                                                                     \\ \hline
		$ T $                                                                     & \multicolumn{10}{c|}{30}                                                                                                                                                                                                                             \\ \hline
		Offset $ s $                                                              & \multicolumn{10}{c|}{0.008}                                                                                                                                                                                                                          \\ \hline
		Training iterations                                                       & \multicolumn{10}{c|}{6000}                                                                                                                                                                                                                           \\ \hline
		Initial learning rate                                                     & \multicolumn{10}{c|}{0.001}                                                                                                                                                                                                                          \\ \hline
		Final learning rate                                                       & \multicolumn{1}{c|}{1e-4} & \multicolumn{9}{c|}{1e-5}                                                                                                                                                                                                \\ \hline
	\end{tabular}
\end{table*}

\begin{algorithm}[H]
	\caption{QD3PM generation.}  
	\label{alg:generation}
	\begin{algorithmic}[1] 
		\STATE \textbf{Input:} Number of qubits $N$, total timestep of the diffusion process $T$, the optimal parameters $\boldsymbol{\omega}^{*}, \boldsymbol{v}^{*}$.
		\STATE Initialize $\boldsymbol{x}_t \sim q(\boldsymbol{x}_T)$ \quad $ \# $ sample from a uniform distribution  
		\FOR{$t$ from $T$ to 1}
		\STATE $ |\psi(t, \boldsymbol{x}_t)\rangle \leftarrow R_y\left(\frac{2\pi \cdot t}{T+1}\right)|0\rangle \otimes_{i=1}^{N} R_y\left(w^{*}_i \cdot x_t^{(i)}\right)|0\rangle $
		\STATE $ |\psi_{\text{out}}\rangle \leftarrow U(\boldsymbol{v}^{*})|\psi(t, \boldsymbol{x}_t)\rangle $
		\STATE $ \rho_{\text{out}} \leftarrow \text{Tr}_1(|\psi_{\text{out}}\rangle\langle\psi_{\text{out}}|) $
		\STATE Measure $ \rho_{\text{out}} $ in the computational basis and obtain a measure outcome $\boldsymbol{x}_t \sim p_{\boldsymbol{\theta}}(\boldsymbol{x}_{t-1}| \boldsymbol{x}_t )=\text{Tr}(|\boldsymbol{x}_{t-1}\rangle\langle\boldsymbol{x}_{t-1}|\rho_{\text{out}}) $
		\ENDFOR
		\STATE \textbf{Output:} The final generation data $ \boldsymbol{x}_{0}$
	\end{algorithmic}  
\end{algorithm}

\subsection{Simulation Setup}
\label{section: simulations setup}
We implement our model using TensorFlow \cite{abadi2016tensorflow} and optimize parameters with the Adam optimizer \cite{kingma2014adam}. Quantum circuit simulations are performed with the Tensorcircuit framework \cite{zhang2023tensorcircuit}. To improve training, we use a cosine learning rate decay strategy \cite{loshchilov2016sgdr} for dynamic adjustment of the learning rate. All simulations are repeated with 5 random seeds to gather statistical features. For the BAS dataset, we assume uniform probability across all patterns. The performance of the 4-bit BAS dataset under different qubit topologies is shown in \cref{section: other connectivity}.
We use the Gaussian initial method \cite{zhang2022escaping} to initialize the trainable parameters in the quantum circuit before training. The hyperparameters are shown in Tables \ref{tab: hyperparameters1} and \ref{tab: hyperparameters2}. Our choice of these parameters is guided by the principle that to avoid overfitting in quantum generative models, it is essential to enhance model expressivity via greater circuit depth and to provide sufficient training data \cite{gili2023quantum}.

We note that the primary objective of this work is to establish a solid theoretical framework for improving discrete diffusion models and showcasing a quantum advantage. Consequently, our simulations are conducted under noiseless conditions, leaving the investigation of the model's performance in noisy environments for future research.
% Our code is available at: \hyperlink{https://github.com/ChuangtaoChen/QD3PM}{https://github.com/ChuangtaoChen/QD3PM}.

The data in Mixed Gaussian datasets follows a mixed Gaussian distribution: $$	\pi(x) \propto e^{-(1/2)((x - \mu_1)/\nu)^2} + e^{-(1/2)((x - \mu_2)/\nu)^2},$$
where $ x $ takes integer values (encoded into quantum states as bitstrings) from 1 to $ x_{\text{max}} $, with $ x_{\text{max}} = 2^N $. We set $ \nu = \frac{1}{8}x_{\text{max}} $, and the distribution centers are $ \mu_1 = \frac{2}{7}x_{\text{max}} $ and $ \mu_2 = \frac{5}{7}x_{\text{max}} $. 
% A total of 50,000 samples are drawn from $ \pi(x) $ to form the training dataset $ \mathcal{D}_{\text{train}} $. 

\begin{table}[h]
	\centering
	\caption{The hyperparameters of training circuit as $p_{\boldsymbol{\theta}}(\boldsymbol{x}_{0}|\boldsymbol{x}_t)$.}  
	\label{tab: hyperparameters2}  
	\begin{tabular}{|c|ccccc|}
		\hline
		$ N $                     & \multicolumn{1}{c|}{4}  & \multicolumn{1}{c|}{6}  & \multicolumn{1}{c|}{8}   & \multicolumn{1}{c|}{9}   & 10  \\ \hline
		$L$                       & \multicolumn{1}{c|}{14} & \multicolumn{1}{c|}{20} & \multicolumn{1}{c|}{22}  & \multicolumn{1}{c|}{25}  & 30  \\ \hline
		Batch size                & \multicolumn{1}{c|}{16} & \multicolumn{1}{c|}{32} & \multicolumn{1}{c|}{128} & \multicolumn{1}{c|}{256} & 256 \\ \hline
		Learning rate decay steps & \multicolumn{5}{c|}{1000}                                                                                     \\ \hline
		$ T $                     & \multicolumn{5}{c|}{30}                                                                                       \\ \hline
		Offset $ s $              & \multicolumn{5}{c|}{0.008}                                                                                    \\ \hline
		Training iterations       & \multicolumn{5}{c|}{6000}                                                                                     \\ \hline
		Initial learning rate     & \multicolumn{5}{c|}{0.001}                                                                                    \\ \hline
		Final learning rate       & \multicolumn{5}{c|}{1e-5}                                                                                     \\ \hline
	\end{tabular}
\end{table}

\section*{DATA AVAILABILITY}

Data is available from the corresponding author upon reasonable request.

\section*{CODE AVAILABILITY}

The codes used to generate data for this paper are available from the corresponding author upon reasonable request.

\section*{ACKNOWLEDGEMENTS}
This work is supported by the Science and Technology
Development Fund, Macau SAR (0093/2022/A2,
0076/2022/A2, and 0008/2022/AGJ), National Natural Science Foundation of China (62471187)
Guangdong Basic and Applied Basic Research Foundation (Grant Nos. 2025A1515011489, 2022A1515140116). 
%Chuangtao Chen thanks Xuefen Zhang for her helpful discussions.

\section*{AUTHOR CONTRIBUTIONS}
C.C. conceived the idea, designed the methodology, proved Theorems, implemented the code, and wrote the original draft. Q.Z. performed the simulations and supervised the project. M.Z. established the simulation environment and conducted data analysis and visualization. D.N. developed the code for the circuit model and prepared figures. Z.H. assisted with the simulations and provided the conceptual basis for the proof of Theorem 1. H.S. assisted with the proof of Theorem 2, helped design the posterior state circuit and implement the simulations. All authors contributed to the review and editing of the final manuscript.

\section*{COMPETING INTERESTS}
The authors declare no competing interests.

\onecolumngrid
%\section*{References}
%\bibliographystyle{unsrt}
\bibliography{example_paper}

%%%%%%%%%%%%%%%%%%%%%%%%%%%%%%%%%%%%%%%%%%%%%%%%%%%%%%%%%%%%

\newpage

\allowdisplaybreaks
\appendix

\section{Related Works}
\label{section2}
\subsection{Quantum Generative Models}
Quantum generative models \cite{tian2023recent} leverage quantum computing to redesign novel generative frameworks, including quantum generative adversarial networks (QGANs) \cite{lloyd2018quantum, dallaire2018quantum,chakrabarti2019quantum,zoufal2019quantum,situ2020quantum,huang2021quantum,niu2022entangling,silver2023mosaiq,kim2024hamiltonian,ma2025quantum}, quantum circuit Born machines \cite{liu2018differentiable,benedetti2019generative}, quantum variational autoencoders \cite{khoshaman2018quantum,wang2024zeta}, quantum Boltzmann machines \cite{amin2018quantum,zoufal2021variational,coopmans2024sample,minervini2025evolved}, and other approaches \cite{kyriienko2024protocols,wu2024multidimensional}.

\textbf{Quantum Generative Adversarial Networks.} 
The study of QGANs began with foundational works in 2018. Lloyd et al. \cite{lloyd2018quantum} first introduce the theoretical concept, showing that a quantum generator and discriminator could converge to a unique fixed point and potentially offer an exponential advantage over classical methods. Concurrently, Dallaire-Demers et al. \cite{dallaire2018quantum} propose a practical implementation using parameterized quantum circuits and outlined a quantum method for computing the required gradients, demonstrating its feasibility through numerical experiments. Subsequently, Chakrabarti et al. \cite{chakrabarti2019quantum} introduce the Quantum Wasserstein GAN (QWGAN), leveraging a Wasserstein semimetric to enhance robustness and scalability, particularly on noisy hardware. Zoufal et al. \cite{zoufal2019quantum} demonstrate the use of QGANs for learning and loading complex random distributions. Subsequently, Situ et al. \cite{situ2020quantum} design a hybrid quantum-classical model specifically for generating discrete distributions, addressing the known vanishing gradient problem in classical GANs.
In 2021, Huang et al. \cite{huang2021quantum} report the implementation of a QGAN on multiple superconducting qubits, successfully training the network to a Nash equilibrium point where it could replicate data from a training set. Niu et al. \cite{niu2022entangling} introduce the Entangling QGAN (EQ-GAN), a novel architecture that performs entangling operations between the generated and true data to overcome prior limitations and mitigate errors, demonstrating it on a Google Sycamore processor. In 2023, Silver et al. \cite{silver2023mosaiq} propose the MosaiQ framework to address the poor quality and robustness of images generated on NISQ computers. In 2024, Kim et al. \cite{kim2024hamiltonian} introduce Hamiltonian QGANs (HQuGANs), which use quantum optimal control instead of circuits to better align with hardware constraints. Recently, Ma et al. \cite{ma2025quantum} develope a QGAN based on Quantum Implicit Neural Representations (QINR), which demonstrated the ability to generate high-resolution images with a significant reduction in the number of required trainable quantum parameters compared to previous state-of-the-art models.

% \textcolor{blue}{\textbf{Quantum Generative Adversarial Networks.} QGANs are widely studied and applied quantum generative models, drawing inspiration from classical GANs \cite{goodfellow2014generative} and adapted for quantum generative tasks through partial or full quantization.}

% \textbf{quantum Variational Autoencoders.}

%\textbf{quantum Boltzmann Machines.}

\textbf{Quantum Circuit Born Machines.}
Foundational work in this area was pioneered by Benedetti et al. \cite{benedetti2019generative}, who propose a versatile quantum circuit learning algorithm for generative tasks and introduced the hardware-independent qBAS score for benchmarking. Around the same time, Liu et al. \cite{liu2018differentiable} address the challenge of intractable likelihoods by introducing an efficient gradient-based training algorithm using a kerneled maximum mean discrepancy loss. Coyle et al. \cite{coyle2020born} demonstrate that using Sinkhorn divergence and Stein discrepancy as cost functions outperforms previous methods and strengthens the case for quantum advantage. The practical scope of these models was expanded by Kiss et al. \cite{kiss2022conditional}, who applied Born machines to generate complex conditional distributions for Monte Carlo simulations in high-energy physics on real quantum hardware. Recently, Gili et al. \cite{gili2023introducing} introduce non-linearity into the model with the Quantum Neuron Born Machine to better learn challenging distributions and later pioneered the evaluation of a QCBM's generalization capability to generate novel, high-quality data \cite{gili2023quantum}. To address the core data loading problem, Li et al. \cite{li2025adaptive} propose an adaptive framework that dynamically grows the circuit ansatz to more efficiently encode complex, real-world data.

\textbf{Quantum Denoising Diffusion Probabilistic Models.} 
Recent advances have reimagined denoising diffusion probabilistic models (DDPMs) in quantum frameworks. Parigi et al. \cite{parigi2024quantum} propose three quantum-noise-driven diffusion models (QNDGDMs) that harness quantum coherence, entanglement, and inherent quantum processor noise to generate complex distributions efficiently, with numerical simulations suggesting potential advantages over classical sampling methods. Zhang et al. \cite{zhang2024generative} introduce the Quantum Denoising Diffusion Probabilistic Model (QuDDPM), which utilizes quantum scrambling for the forward noisy diffusion process and quantum measurement within layered circuits for the backward denoising process, enabling efficient generative learning of quantum data. Chen et al. \cite{chen2024quantum} develop a quantum generative diffusion model (QGDM) that transforms quantum states into completely mixed states via a non-unitary forward process and reconstructs them using a parameter-efficient backward process, achieving 53\% higher fidelity in mixed-state generation than quantum GANs. Francesca et al. \cite{de2024quantuma} demonstrate hybrid quantum-classical diffusion models by embedding variational quantum circuits into classical U-Net architectures (via Quantum ResNet blocks or encoder-level hybridization), showing improved image quality, accelerated convergence, and reduced parameter counts compared to purely classical counterparts. Their subsequent work \cite{de2024quantumb} further propose a quantum latent diffusion model utilizing classical autoencoders and variational circuits, which outperforms classical models in image generation metrics. Kwun et al. \cite{kwun2024mixed} propose a mixed-state quantum denoising diffusion model (MSQuDDPM) that removes the need for scrambling unitaries by integrating depolarizing noise channels and parameterized circuits with projective measurements, overcoming the challenge of implementing high-fidelity random unitary operations in existing methods.

While these quantum adaptations of diffusion models represent a nascent yet promising direction, current implementations remain underexplored for generating classical discrete data, which is a key motivation for our work. Although quantum-classical hybrid approaches exhibit competitive generative performance, they partially underutilize quantum computational advantages. Existing studies have laid the foundational groundwork but lack rigorous theoretical integration of quantum computing principles with diffusion mechanisms, particularly in deriving the exact form of the ground-truth posterior distribution. Our work addresses this gap by establishing a solid theoretical foundation that bridges quantum computation and diffusion models and demonstrating quantum advantage through both theoretical analysis and numerical experiments.

\subsection{Discrete Diffusion Models}
The theoretical foundation for discrete-state diffusion models traces back to Sohl-Dickstein et al. \cite{Sohl2015}, who pioneered binary-variable diffusion processes. 
Hoogeboom et al. \cite{hoogeboom2021argmax} introduce continuous-to-categorical mappings via argmax functions and multinomial noise diffusion for improved text/image segmentation modeling. Austin et al. \cite{austin2021structured} propose Discrete Denoising Diffusion Probabilistic Models (D3PM) 
by designing categorical corruption processes that avoid continuous embeddings and incorporate domain knowledge into transition matrices.
Campbell et al. \cite{campbell2022continuous}  introduce a complete continuous-time framework using Continuous Time Markov Chains for denoising diffusion models on discrete data, enabling efficient training, high-performance sampling, and providing a theoretical bound on the generated distribution error.
Chen et al. \cite{chenanalog} generates discrete data by modeling bits as analog values with self-conditioning and asymmetric time intervals to boost image/caption generation quality.
Santos et al. \cite{santos2023blackout} develop a theoretical framework for diffusion models on discrete-state spaces using arbitrary discrete-state Markov processes for the forward diffusion, deriving the corresponding reverse-time processes and applying the framework to ``Blackout Diffusion\textquotedblright ~for generating samples from an empty state.
Lou et al. \cite{loudiscrete} bridge discrete diffusion with score entropy loss achieving language modeling competitive with GPT-2 and controllable infilling without annealing.
Hayakawa et al. \cite{hayakawa2024distillation} propose Di4C, which addresses the slow sampling issue in discrete diffusion models by introducing a mixture model architecture and novel distillation loss functions that enable student models to capture dimensional correlations and distill multi-step teacher models into few-step samplers while maintaining generation quality.
Liu et al.~\cite{liu2025discrete} identify that the factorized denoising assumption in discrete diffusion models introduces an irreducible total correlation term into the ELBO, and propose Discrete Copula Diffusion (DCD) to incorporate copula-based dependency modeling at inference time, achieving comparable or improved text generation quality with 8 to 32 times fewer denoising steps.

\section{Brief Review of Quantum Conditional States Formalism}
\label{section:Quantum Conditional States Formalism}
In this appendix, we briefly introduce the knowledge of quantum conditional state formalism required to understand our work. To fully comprehend this paper, one should be familiar with at least the concepts introduced below. For a detailed and comprehensive knowledge base, please refer to the original text \cite{leifer2013towards}.
\cref{table: Review Bayesian} outlines the correspondence between classical Bayesian concepts and their quantum counterparts, including states, joint states, marginalization, conditional states, Bayes' theorem, and belief propagation.

\begin{table}[h]
	\centering
	\caption{Comparing the classical Bayesian inference theory \cite{bertsekas2008introduction} with the conditional states formalism in quantum theory \cite{leifer2013towards}.}
	\begin{tabular}{|c|c|c|}
		\hline
		& Classical             & Quantum               \\ \hline
		State                                                          & $P(R)$                     & $\rho_A$                    \\ \hline
		Joint state                                                    & $P(R,S)$                    & $\varrho_{AB}$                    \\ \hline
		Marginalization                                                & $P(S) = \sum_R P(R,S)$                    & $\rho_B = \text{Tr}_A(\varrho_{AB})$                   \\ \hline
		Conditional state                             & $P(S|R)$                  & $\varrho_{B|A}$                    \\
		& $\sum_S P(S|R) = 1$ & $\text{Tr}_B(\varrho_{B|A}) = \boldsymbol{I}_A$ \\ \hline
		Relation between joint  & $P(R,S) = P(S|R)P(R)$                 & $\varrho_{AB} = \varrho_{B|A} \star \rho_A$                 \\
		and conditional states & $P(S|R) = P(R,S)/P(R)$ & $\varrho_{B|A} = \varrho_{AB} \star \rho_A^{-1}$ \\ \hline
		Belief propagation                                             & $P(S) = \sum_R P(S|R)P(R)$                  & $\rho_{B} = \text{Tr}_A(\varrho_{B|A} \rho_A)$                  \\ \hline
		Bayes' theorem                                                 & $P(R|S) = P(S|R)P(R)/P(S)$                  & $\varrho_{A|B} = \varrho_{B|A}\star(\rho_A \rho_B^{-1})$                 \\ \hline
		Bayes conditional  & $P(R|S=s) = \frac{P(S = s|R) P(R)}{P(S = s)}$                  & $\varrho_{A|X=x} = \varrho_{X=x|A}\star(\rho_A \rho_{X=x}^{-1})$                 \\ \hline
	\end{tabular}
    \label{table: Review Bayesian}
\end{table}

\subsection{Basic Concepts}

The conditional state formalism, proposed by Leifer and Spekkens \cite{leifer2013towards}, aims to extend classical Bayesian reasoning into quantum theory. In this framework, classical variables are represented by later letters in the alphabet, such as $ R, S, T, X, Y, Z $, while quantum systems are represented by earlier letters, such as $ A, B, C $.

In conditional state formalism, systems are referred to as ``region\textquotedblright~rather than the traditional ``systems\textquotedblright. A basic region describes a system at a fixed time point, while a composite region consists of multiple disjoint basic regions. For example, the input $ \rho_A $ and output $ \rho_B $ of a quantum channel, which in traditional formalism are treated as states at different time points of the same system, are treated as two independent regions in the conditional state formalism, where a joint region $ \mathcal{H}_{AB} = \mathcal{H}_A \otimes \mathcal{H}_B $ is formed.

In classical Bayesian inference, the joint distribution $ P(R, S) $ describes our knowledge, information, or belief about variables $ R $ and $ S $. Similarly, in the quantum case, the joint state $ \rho_{AB} $ is defined on the Hilbert space $ \mathcal{H}_A \otimes \mathcal{H}_B $ and represents our knowledge about regions $ A $ and $ B $. Quantum marginalization is achieved through the partial trace operation, analogous to summing over variables in classical probability:
\begin{equation}
	\rho_B = \text{Tr}_A(\rho_{AB}), 
\end{equation}
which corresponds to $ P(S) = \sum_R P(R, S) $ in classical probability.

\subsection{Causal Correlations}

The conditional state formalism introduces causality into the quantum state structure. If there is no direct or indirect causal connection between regions $ A $ and $ B $ (for example, if two independent systems are at the same time point), they are causally independent. The joint state $ \rho_{AB} $ represents the regions $ A $ and $ B $ in this case. Both $ \rho_{AB} $ and the conditional states $ \rho_{A|B} $ and $ \rho_{B|A} $ are positive definite. Note that even if regions $ A $ and $ B $ are causally independent, they can still exhibit statistical correlations due to entanglement effects \cite{leifer2013towards}.

Conversely, if there is a direct or indirect causal influence from region $ A $ to region $ B $, such as when the same system is measured at different time points (i.e., $ A $ is the input of a quantum channel and $ B $ is the output), then $ A $ and $ B $ are causally related, as changes in $ \rho_A $ will lead to changes in $ \rho_B $. In this case, the joint state is represented by an operator $ \varrho_{AB} $ on the space $ \mathcal{H}_{AB} = \mathcal{H}_A \otimes \mathcal{H}_B $.

%Operators $ \varrho_{B|A} $ and $ \varrho_{A|B} $ must also have a positive partial transpose, referred to as causal states.

This distinction ensures that dynamical processes (such as those represented by Completely Positive Trace-preserving (CPT) maps in quantum channels) are properly described in the formalism. The operator $ \varrho_{B|A} $ describes the process from $ A $ to $ B $. Specifically, if the dynamics in the traditional formalism are described by a completely positive trace-preserving (CPT) map $ \mathcal{E}_{B|A}: \mathfrak{L}(\mathcal{H}_A) \rightarrow \mathfrak{L}(\mathcal{H}_B) $, the corresponding conditional state $ \varrho_{B|A} $ is related to $ \mathcal{E}_{B|A} $ via the Choi–Jamiołkowski isomorphism \cite{choi1975completely,jamiolkowski1972linear}:
\begin{equation}
	\varrho_{B|A} =  (\mathcal{E}_{B|A^{'}} \otimes \boldsymbol{I}_{A} )(|\tilde{\Phi}^{+}\rangle\langle\tilde{\Phi}^{+}|_{A^{'}A}),
\end{equation} 
where $ \mathcal{H}_{A'} $ is isomorphic to $ \mathcal{H}_A $, and $ |\tilde{\Phi}^{+}\rangle = \sum_{\boldsymbol{i}=0}^{d-1} |\boldsymbol{i}\rangle \otimes |\boldsymbol{i}\rangle $ is the unnormalized maximally entangled state used to simplify the derivations and emphasize the mathematical form. In our work, we use the normalized maximally entangled state to derive the posterior state, ensuring the physical interpretability of the process and results.

The conditional state $ \varrho_{B|A} $ follows quantum belief propagation rules:
\begin{equation}
	\rho_B = \text{Tr}_A(\varrho_{B|A} \rho_A).
\end{equation}
Similar to the classical belief propagation rule:
\begin{equation}
	P(S) = \sum_R P(S|R)P(R).
\end{equation}

The conditional state formalism can also describe the state of a classical system. For a random variable $ X $, the probability distribution $ P(X) $ corresponds to a diagonal density operator:
\begin{equation}
	\rho_X = \sum_x P(X=x) |x\rangle \langle x|_X.
\end{equation}

For a quantum state $ \rho_A $ measured to obtain the random variable $ X $, this process corresponds to a mixed operator:
\begin{equation}
	\varrho_{X|A} = \sum_x |x\rangle \langle x|_X \otimes E_x^{A},
\end{equation}
where the set $ \{E_x^{A}\} $ is a POVM on region $ A $, meaning each $ E_x^{A} $ is positive and $ \sum_x E_x^{A} = \boldsymbol{I} $. The Born rule can be expressed as quantum belief propagation in the conditional state formalism:
\begin{equation}
	\rho_X = \text{Tr}_A(\varrho_{X|A} \rho_A).
\end{equation}

\subsection{Composition of Channels}
In certain cases, we are concerned with the causal state representation under multiple channel operations. Specifically, when $ \rho_A $ undergoes the channel $ \mathcal{E}_{B|A} $ to produce $ \rho_B $, and then $ \rho_B $ undergoes the channel $ \mathcal{E}_{C|B} $, we have a joint channel $ \mathcal{E}_{C|A} = \mathcal{E}_{C|B} \circ \mathcal{E}_{B|A} $. The Choi–Jamiołkowski operators for the three channels are $ \varrho_{C|A} $, $ \varrho_{C|B} $, and $ \varrho_{B|A} $, satisfying:
\begin{equation}
	\varrho_{C|A} = \text{Tr}_B(\varrho_{C|B} \varrho_{B|A}).
\end{equation}

\subsection{Quantum Bayes' Theorem}

We introduce two forms of the quantum Bayes' theorem under conditional formalisms: one for quantum-to-quantum regions and one based on classical measurement outcomes for quantum regions.

For two causally related regions, $ \rho_A $ undergoes the CPT map $ \mathcal{E}_{B|A} $ to produce $ \rho_B $, with the isomorphic operator of the channel $ \mathcal{E}_{B|A} $ given by $ \varrho_{B|A} $. The causal version of the quantum Bayes' theorem is then:
\begin{equation}
	\varrho_{A|B} = \varrho_{B|A} \star (\rho_A \rho_B^{-1}),
\end{equation}
where the $ \star $ product is defined as $ M \star N = N^{\frac{1}{2}} M N^{\frac{1}{2}} $.

For the case where $ \rho_A $ is measured to yield a classical random variable $ X $, the quantum Bayes' theorem is:
\begin{equation}
	\varrho_{A|X} = \varrho_{X|A} \star (\rho_A \rho_X^{-1}).
\end{equation}

\subsection{Quantum Bayesian Conditioning}

The conditional state formalism extends the Bayesian conditional to the quantum domain. In classical Bayesian conditioning, when we focus on a random variable $ R $ and learn the value of a related variable $ X = x $, we update the probability distribution of $ R $ from the prior $ P(R) $ to the posterior $ P(R|X=x) $. Quantum Bayesian conditioning concerns how to update our understanding of causally related regions when we learn the specific value of a related variable. When there is a mixed state $ \rho_{XA} $ and we observe $ X = x $, the quantum state $ \rho_A $ should be updated as:
\begin{equation}
	\rho_{A|X=x} = \varrho_{X=x|A} \star (\rho_A \rho_{X=x}^{-1}).
\end{equation}

\section{More Results}
\label{section: more result}

\begin{figure}[h]
	\centering
	\includegraphics[width=0.5\linewidth]{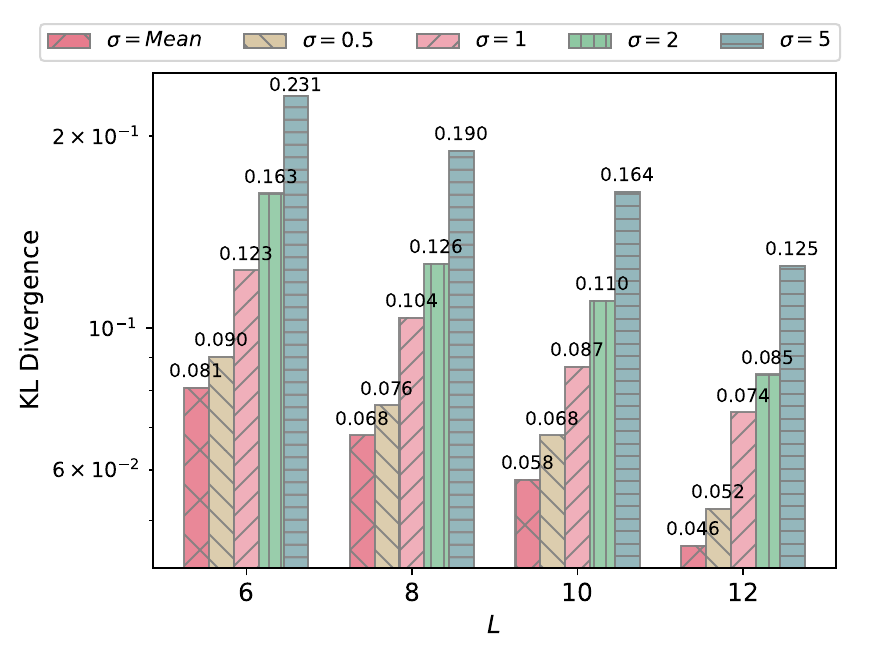}
	\caption{KL divergence for different bandwidths \( \sigma \) and circuit layers \( L \), showing the model's fitting performance relative to the target distribution .}
	\label{fig:4BAS_sigma_L}
\end{figure}

\subsection{Impact of Bandwidth Parameter}
\label{section: Bandwidth}
We evaluate the impact of the bandwidth parameter \( \sigma \) in the MMD Gaussian kernel and the denoising circuit layer \( L \) on the performance of the 4-bit BAS dataset generation. 
The bandwidth \( \sigma \) significantly influences the model's trainability and its ability to fit the underlying data distribution \cite{rudolph2024trainability}. 
We conduct experiments with \( \sigma \) values of \( 0.5, 1, 2, 5 \), and use the mean of multiple bandwidths, i.e., \( \sigma = [0.01, 0.1, 0.25, 0.5, 1.0, 10]N \). Additionally, we set the denoising circuit layer to \( L = 6, 8, 10 \), and 12.

\cref{fig:4BAS_sigma_L} displays the KL divergence after training for different \( \sigma \) and \( L \) settings, measuring the difference between the trained model's generated distribution and the target distribution. As \( \sigma \) increases, the KL divergence rises, indicating worse fitting performance. Among all configurations, using the mean bandwidth yields the best fitting, with the smallest KL divergence. Moreover, increasing the number of quantum circuit layers reduces the KL divergence, improving the model’s fitting accuracy and overall generative performance. This is because a deeper circuit possesses greater expressive power, enabling it to more effectively approximate the target distribution.

\subsection{Results of 4-bit BAS with Other Connectivity}
\label{section: other connectivity}
In the main text of our paper, we assume that the qubits in the denoising circuit are all-to-all connected. However, most near-term quantum computers do not follow this topological connection scheme. In this section, we report the performance of a 4-bit BAS experiment under both chain and star connections. In the chain connection, neighboring qubits are allowed to connect, forming a linear structure like $q1-q2-q3-q4$. In the star connection, the first qubit is connected to all the other qubits, but the other qubits are not interconnected.

\cref{fig:4BAS star loss} shows the trend of the loss function under the star connection in the denoising circuit. \cref{fig:4BAS star generation}(a) illustrates the change in the discrepancy between the model distribution and the target one throughout the training process under the star connection. \cref{fig:4BAS star generation}(b) presents the QD3PM sampling distribution after training with $L=12$. It can be observed that most of the generated samples belong to the BAS category, but there is also a portion of samples of other types, distributed uniformly. This suggests that the star connection topology performs worse than the all-to-all connection, likely due to insufficient information fusion between the qubits encoding the data.

\begin{figure}[H]
	\centering
	\includegraphics[width=0.45\linewidth]{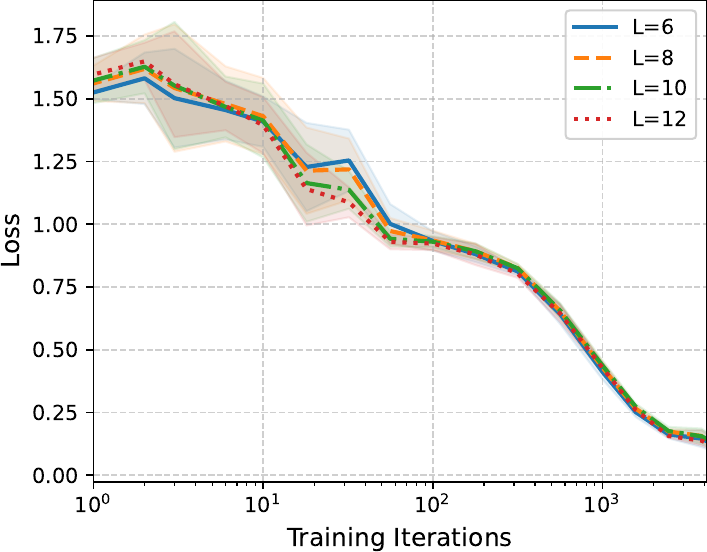}
	\caption{Training loss evolution in 4-bit BAS generation with star connectivity.}
	\label{fig:4BAS star loss}
\end{figure}
\begin{figure}[H]
	\centering
	\includegraphics[width=0.5\linewidth]{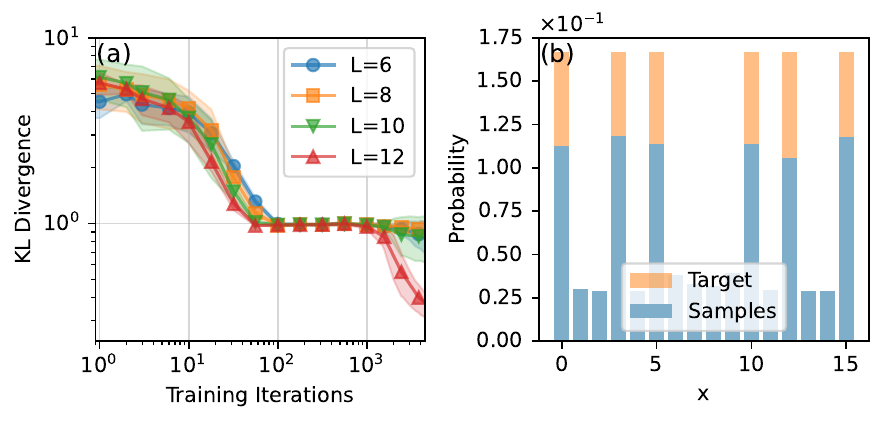}
	\caption{(a) KL divergence, and (b) generated samples for the 4-bit BAS dataset with star connectivity.}
	\label{fig:4BAS star generation}
\end{figure}

\cref{fig:4BAS chain loss} shows the training loss evolution for 4-bit BAS synthesis with chain connectivity. The loss decreases progressively as training iterations increase for all $ L $ values ($ L=6 $, 8, 10, and 12), indicating improved model performance. 
\cref{fig:4BAS chain generation}(a) displays the reduction of KL divergence over training iterations, suggesting that the model's training is improving as the iterations increase. \cref{fig:4BAS chain generation}(b), the probability distribution of the target and the generated samples are compared. The QD3PM with chain connectivity is unable to fully fit the BAS pattern, which may be due to insufficient information exchange between qubits under the chain topology. Designing efficient denoising circuit structures for QD3PM is an interesting and important research direction, which we leave for future work. An effective approach to this problem could be using quantum architecture search algorithms \cite{du2022quantum,situ2024distributed,he2024training,furrutter2024quantum,he2025self,Su2025topology}, which streamline the design process by automating the discovery of optimal circuits. This presents a more scalable and robust alternative to relying on manual, intuition-based design.

\begin{figure}[H]
	\centering
	\includegraphics[width=0.45\linewidth]{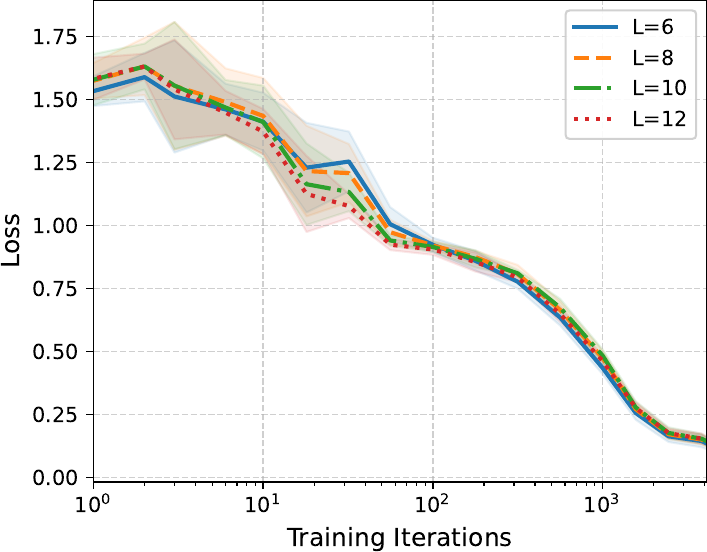}
	\caption{Training loss evolution in 4-bit BAS synthesis with chain connectivity.}
	\label{fig:4BAS chain loss}
\end{figure}
\begin{figure}[H]
	\centering
	\includegraphics[width=0.5\linewidth]{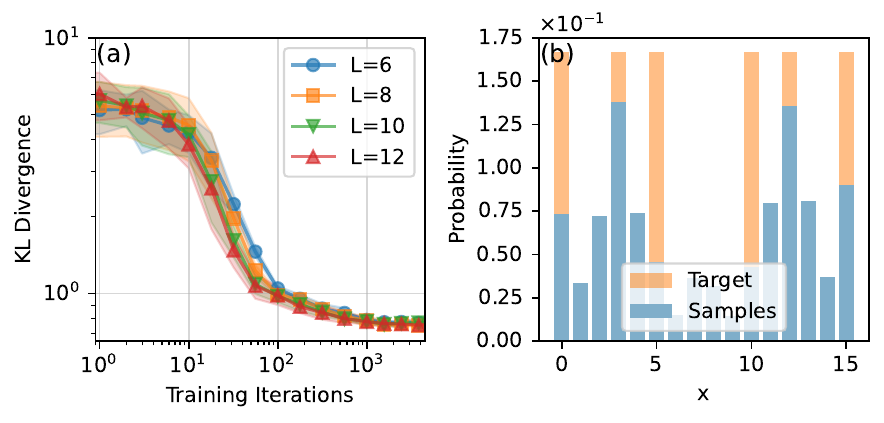}
	\caption{(a) KL divergence, and (b) generated samples for the 4-bit BAS dataset with chain connectivity.}
	\label{fig:4BAS chain generation}
\end{figure}

\subsection{Training Loss and Generation Dynamics}

\subsubsection{Training the Circuit for Modeling $p_{\boldsymbol{\theta}}(\boldsymbol{x}_{t-1}|\boldsymbol{x}_t)$}

\begin{figure}[H]
	\centering
	\subfloat[$N=4$\label{subfig:a}]{
		\includegraphics[width=0.18\textwidth]{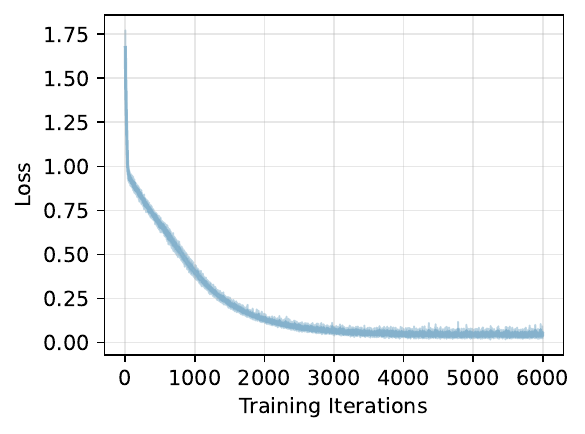}
	} \hfill
	\subfloat[$N=6$\label{subfig:b}]{
		\includegraphics[width=0.18\textwidth]{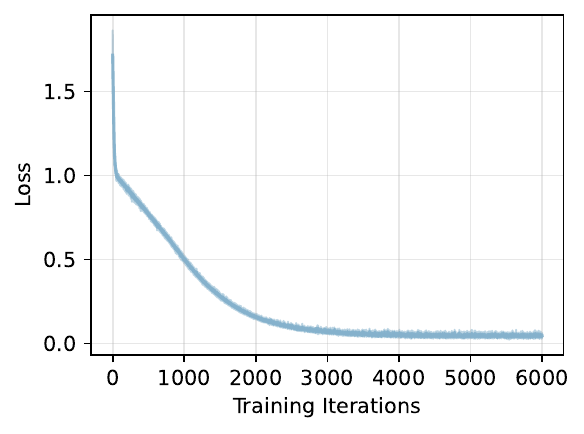}
	} \hfill
	\subfloat[$N=8$\label{subfig:c}]{
		\includegraphics[width=0.18\textwidth]{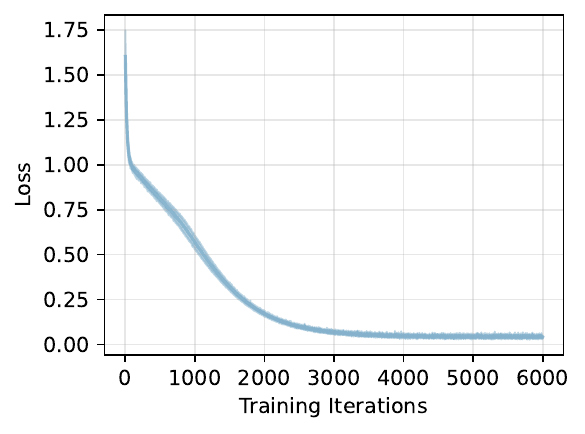}
	} \hfill
	\subfloat[$N=9$\label{subfig:d}]{
		\includegraphics[width=0.18\textwidth]{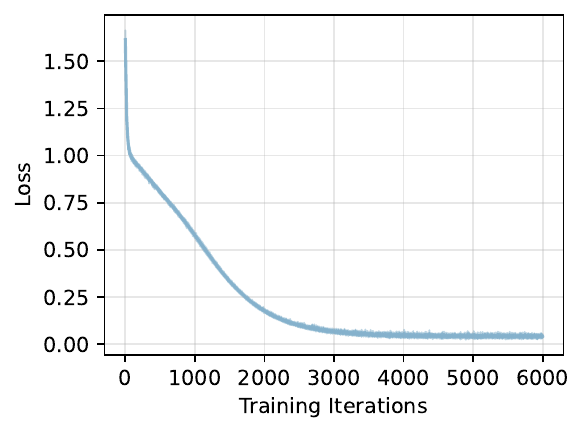}
	} \hfill
	\subfloat[$N=10$\label{subfig:e}]{
		\includegraphics[width=0.18\textwidth]{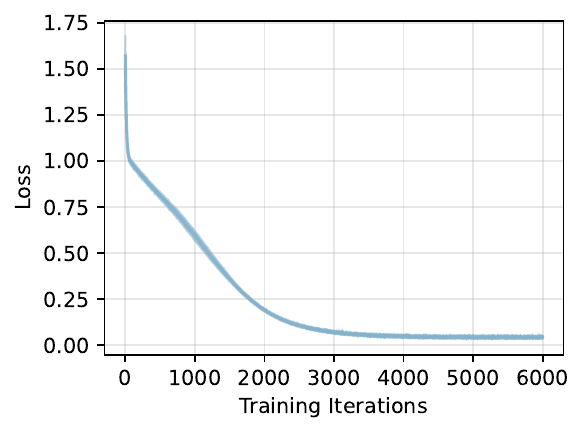}
	}
	\caption{Training loss evolution over iterations for fitting Mixed Gaussian datasets with varying qubit numbers $N$.}
	\label{fig: loss gausssian}
\end{figure}

%%%%%%%%%% Gaussian  %%%%%%%%%%%%%%%%%%%%%%%%%%%
\begin{figure}[H]
	\centering
	\subfloat[$N=4$\label{subfig:a}]{
		\includegraphics[width=0.18\textwidth]{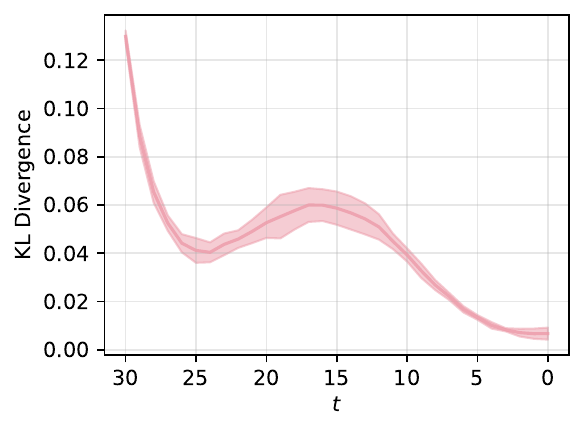}
	} \hfill
	\subfloat[$N=6$\label{subfig:b}]{
		\includegraphics[width=0.18\textwidth]{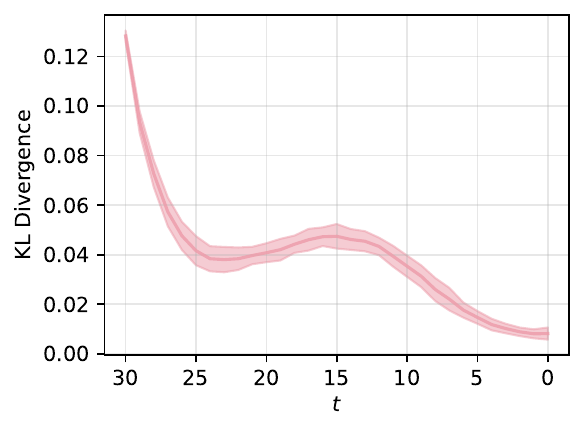}
	} \hfill
	\subfloat[$N=8$\label{subfig:c}]{
		\includegraphics[width=0.18\textwidth]{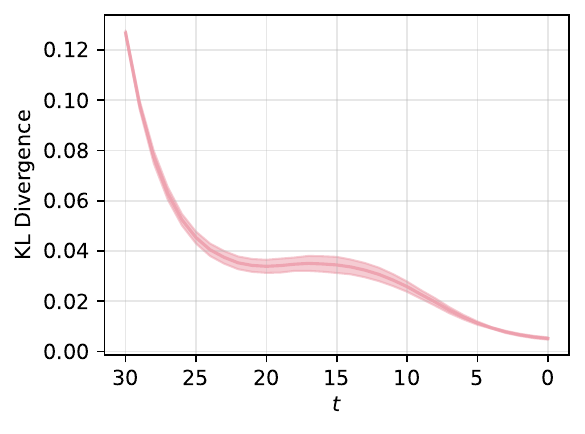}
	} \hfill
	\subfloat[$N=9$\label{subfig:d}]{
		\includegraphics[width=0.18\textwidth]{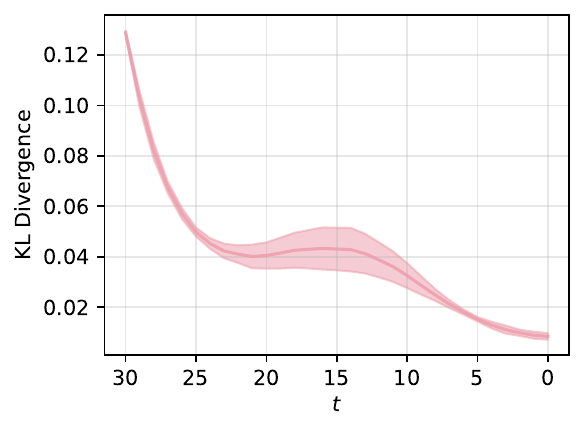}
	} \hfill
	\subfloat[$N=10$\label{subfig:e}]{
		\includegraphics[width=0.18\textwidth]{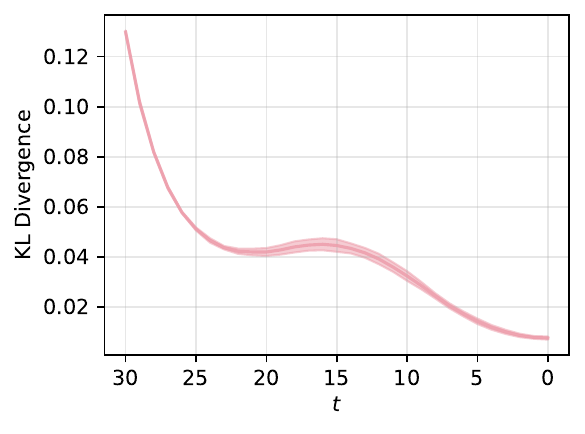}
	}
	\caption{The change in KL divergence between the generated sample distribution and the Mixed Gaussian dataset distribution as timesteps $t$ decrease, after model training.}
	\label{fig: generation kl gaussian}
\end{figure}

%%%%%%%%%% BAS  %%%%%%%%%%%%%%%%%%%%%%%%%%%

\begin{figure}[H]
	\centering
	\subfloat[$N=4$\label{subfig:a}]{
		\includegraphics[width=0.18\textwidth]{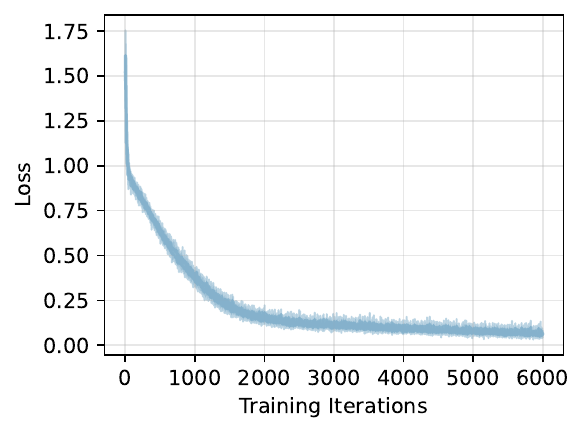}
	} \hfill
	\subfloat[$N=6$\label{subfig:b}]{
		\includegraphics[width=0.18\textwidth]{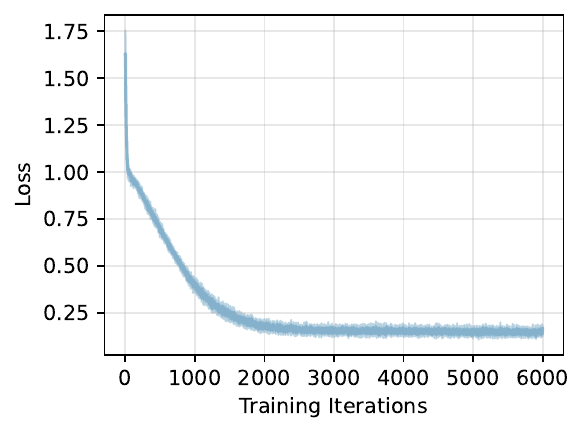}
	} \hfill
	\subfloat[$N=8$\label{subfig:c}]{
		\includegraphics[width=0.18\textwidth]{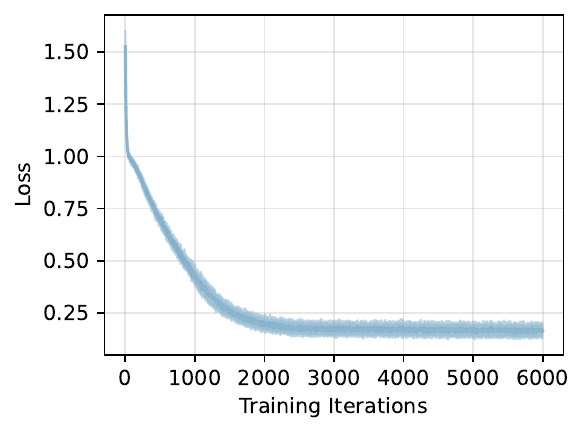}
	} \hfill
	\subfloat[$N=9$\label{subfig:d}]{
		\includegraphics[width=0.18\textwidth]{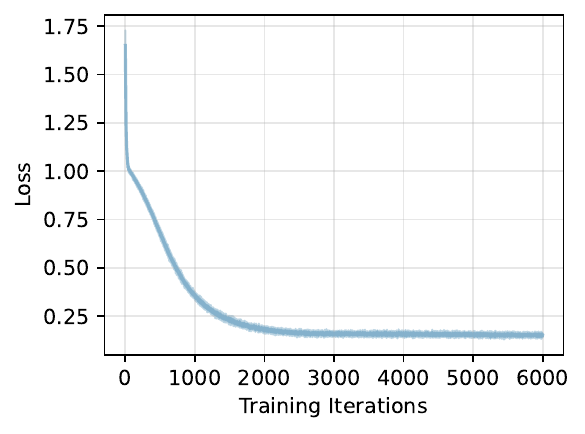}
	} \hfill
	\subfloat[$N=10$\label{subfig:e}]{
		\includegraphics[width=0.18\textwidth]{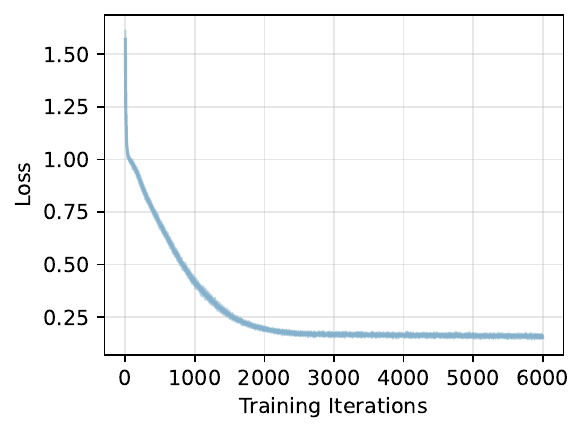}
	}
	\caption{Training loss evolution over iterations for fitting BAS datasets with varying qubit numbers $N$.}
	\label{fig: loss bas}
\end{figure}

\begin{figure}[H]
	\centering
	\subfloat[$N=4$\label{subfig:a}]{
		\includegraphics[width=0.18\textwidth]{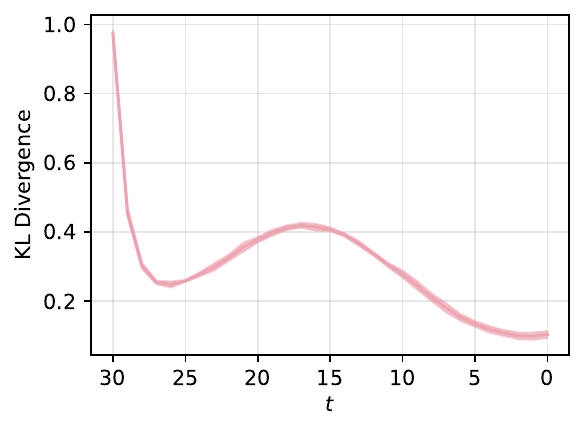}
	} \hfill
	\subfloat[$N=6$\label{subfig:b}]{
		\includegraphics[width=0.18\textwidth]{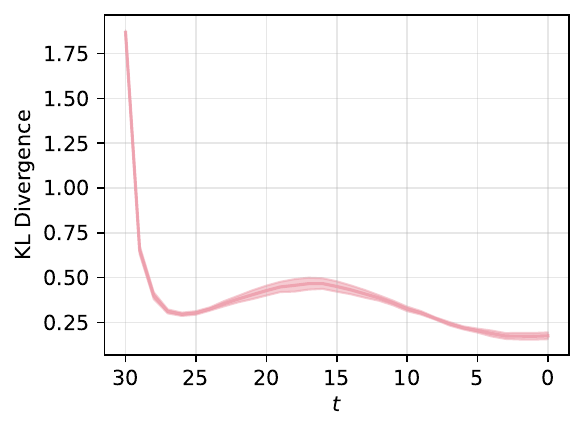}
	} \hfill
	\subfloat[$N=8$\label{subfig:c}]{
		\includegraphics[width=0.18\textwidth]{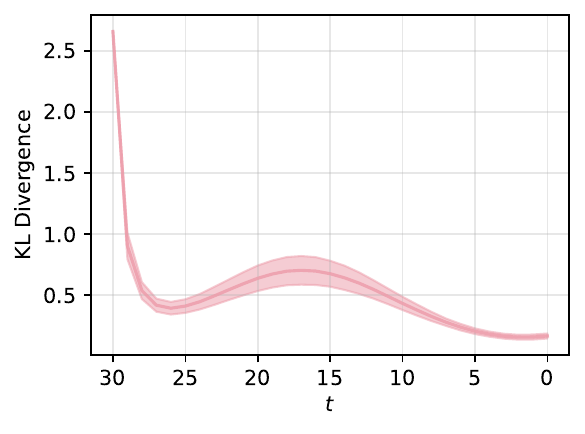}
	} \hfill
	\subfloat[$N=9$\label{subfig:d}]{
		\includegraphics[width=0.18\textwidth]{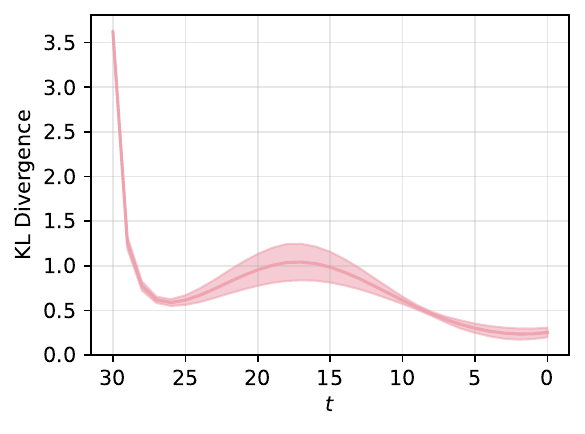}
	} \hfill
	\subfloat[$N=10$\label{subfig:e}]{
		\includegraphics[width=0.18\textwidth]{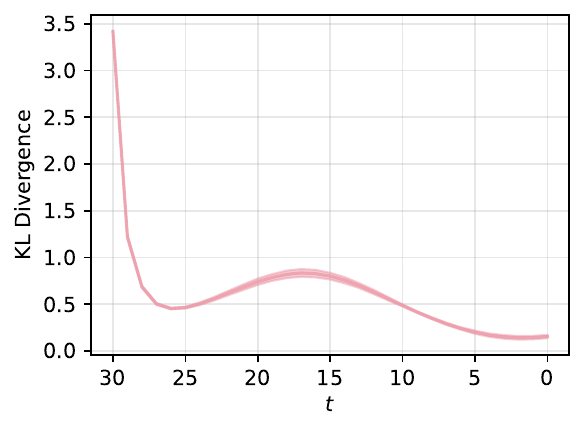}
	}
	\caption{The change in KL divergence between the generated sample distribution and the BAS dataset distribution as timesteps $t$ decrease, after model training.}
	\label{fig: generation kl bas}
\end{figure}

\subsubsection{Training the Circuit for Modeling $p_{\boldsymbol{\theta}}(\boldsymbol{x}_{0}|\boldsymbol{x}_t)$}

\begin{figure}[H]
	\centering
	\subfloat[$N=4$\label{subfig:a}]{
		\includegraphics[width=0.18\textwidth]{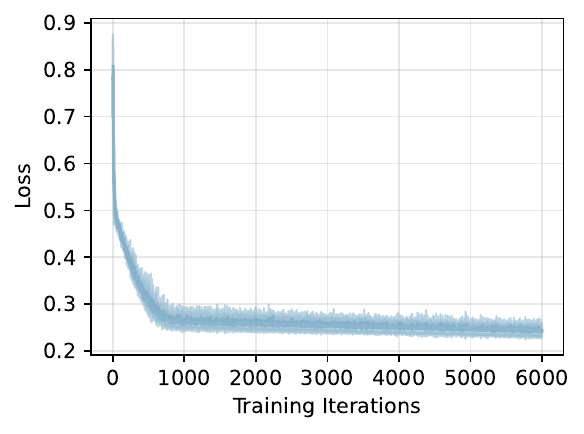}
	} \hfill
	\subfloat[$N=6$\label{subfig:b}]{
		\includegraphics[width=0.18\textwidth]{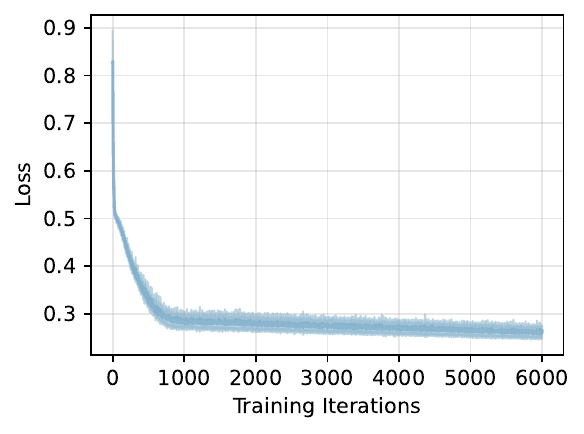}
	} \hfill
	\subfloat[$N=8$\label{subfig:c}]{
		\includegraphics[width=0.18\textwidth]{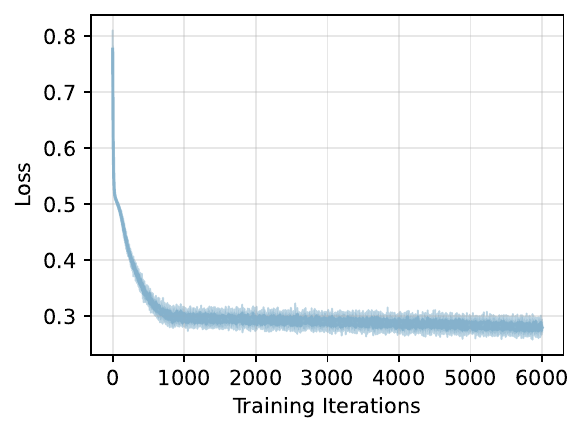}
	} \hfill
	\subfloat[$N=9$\label{subfig:d}]{
		\includegraphics[width=0.18\textwidth]{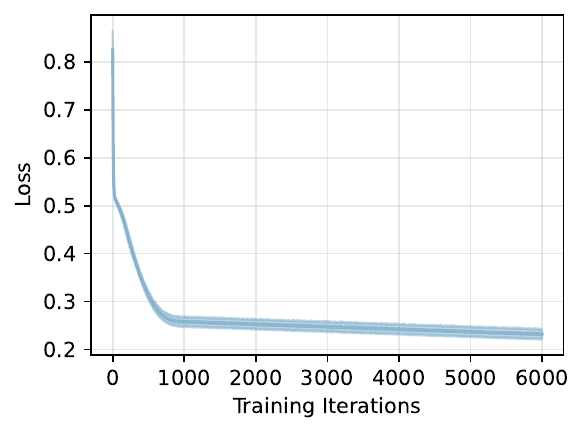}
	} \hfill
	\subfloat[$N=10$\label{subfig:e}]{
		\includegraphics[width=0.18\textwidth]{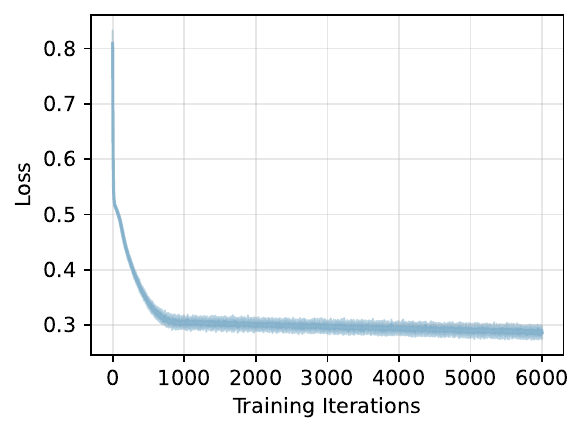}
	}
	\caption{The change in KL divergence between the generated sample distribution and the BAS dataset distribution as timesteps $t$ decrease, after model training.}
	\label{fig: generation kl bas}
\end{figure}

\end{document}